\documentclass[aps,pra,twocolumn,superscriptaddress]{revtex4}

\usepackage{mathpazo}

\usepackage{amsmath}
\usepackage{amsfonts,amssymb,amsthm, bbm}

\usepackage{graphicx}   
\usepackage{amsbsy} 
\usepackage[bold]{hhtensor} 
\usepackage{natbib}
\usepackage{algpseudocode}
\usepackage{xcolor}
\usepackage{comment}
\usepackage{multirow}

\usepackage{epstopdf}

\newcommand\rank{\operatorname{rank}}
\newcommand\diag{\operatorname{diag}}

\newtheorem{definition}{Definition}
\newtheorem{corollary}{Corollary}
\newtheorem{theorem}{Theorem}
\newtheorem{lemma}{Lemma}
\newtheorem{proposition}{Proposition}

\definecolor{myblue}{RGB}{51,51,178}
\definecolor{mygreen}{RGB}{0,128,0}
\definecolor{myred}{RGB}{189,26,26}

\newcommand{\tr}{\mathrm{tr}}

\usepackage[breaklinks=true]{hyperref}

\hypersetup{
  colorlinks   = true, 
  urlcolor     = blue, 
  linkcolor    = blue, 
  citecolor   = green 
}

\begin{document}

\title{Distinguishing quantum states using Clifford orbits}

\author{Richard Kueng}
\author{Huangjun Zhu}
\author{David Gross}
\affiliation{Institute for Theoretical Physics, University of Cologne, Germany}

\date{\today}

\begin{abstract}
It is a fundamental property of quantum mechanics that information is lost as a result of performing measurements.
Indeed, with every quantum measurement one can associate a number -- its \emph{POVM norm constant} -- that quantifies how much the distinguishability of quantum states degrades in the worst case as a result of the measurement.
This raises the obvious question which measurements preserve the most information in these sense of having the largest norm constant.
While a number of near-optimal schemes have been found (e.g.\ the \emph{uniform POVM}, or \emph{complex projective 4-designs}), they all seem to be difficult to implement in practice.
Here, we analyze the distinguishability of quantum states under measurements that are orbits of the Clifford group.
The Clifford group plays an important role e.g.\ in quantum error correction, and its elements are considered simple to implement.
We find that the POVM norm constants of Clifford orbits depend on the \emph{effective rank} of the states that should be distinguished, as well as on a quantitative measure of the ``degree of localization in phase space'' of the vectors in the orbit.
The most important Clifford orbit is formed by the set of stabilizer states. 
Our main result implies that stabilizer measurements are essentially optimal for distinguishing pure quantum states.
As an auxiliary result, we use the methods developed here to prove new entropic uncertainty relations for stabilizer measurements.
This paper is based on a very recent analysis of the representation theory of tensor powers of the Clifford group.
\end{abstract}

\maketitle

\section{Introduction and overview}

\subsection{Distinguishing quantum states}

Helstrom's Theorem \cite{helstrom_quantum_1969} gives a precise measure of the distinguishability of quantum sates.
In the setting of the theorem, one considers a process that prepares one of two states $\rho, \sigma$ with equal probability.
The task is then to estimate which of the two states has been prepared at a given instance.
Helstrom found that the optimal strategy is to perform a two-outcome projective measurement, using the projections onto the non-negative range of $\rho - \sigma$ and onto its complement.
This protocol achieves the optimal probability of identifying the state correctly, which is given by
\begin{equation}
  \mathrm{Pr}_{\textrm{Helstrom}} = \frac{1}{2} + \frac{1}{4} \left\|  \rho -  \sigma \right\|_1,
   \label{eq:helstrom_simple}
\end{equation}
where $\|\rho-\sigma\|_1$ is the \emph{trace norm}, that is, the sum of singular values of the difference.

The formula (\ref{eq:helstrom_simple}) nicely mimics the classical situation. 
Here, one assume that a process first picks -- with equal probability -- one of two distributions $p,q$ over some finite alphabet. 
It then draws one letter of the alphabet according to the chosen distribution.
The task is to decide which of the two distributions has been used.
The optimal strategy is given by the \emph{maximum likelihood rule}, where  one decides for $p$ if the observed sample $x$ is such that $(p-q)(x)$ is positive, and for $q$ otherwise.
The answer is correct  with probability
\begin{equation}
  \mathrm{Pr}_{\textrm{ML}} = \frac{1}{2} + \frac{1}{4} \left\|  p -  q \right\|_{\ell_1},
\end{equation}
where the $\ell_1$-norm of a vector is the sum of the absolute values of its elements.

A slight generalization allows for the two hypotheses to occur with probabilities $(\tau, 1-\tau)$ for $\tau$ not necessarily equal to $\frac12$.
In this case, the expressions generalize to 
\begin{eqnarray*}
  \mathrm{Pr}_{\mathrm{Helstrom}} 
  &=& \frac{1}{2} + \frac{1}{2} \left\| \tau \rho - (1-\tau) \sigma \right\|_1,  \\
  \mathrm{Pr}_{\textrm{ML}} 
  &=& \frac{1}{2} + \frac{1}{2} \left\|  \tau p -  (1-\tau) q \right\|_{\ell_1}.
\end{eqnarray*}
This suggests using the optimal \emph{bias}
\begin{equation*}
	\frac 12 \left\| \tau \rho - (1-\tau) \sigma \right\|_1,
	\qquad
	\frac 12 \left\| \tau p - (1-\tau) q \right\|_{\ell_1} 
\end{equation*}
toward the right solution as a quantitative measure of the distinguishability of states or, respectively, distributions with prior probabilities specified by $\tau$.

Note that the measurement that achieves the quantum bound depends on the particular pair of weighted states $\tau \rho, (1-\tau)\sigma$.
A measurement not optimized to distinguish between those two alternatives may perform considerably worse.
It is thus natural to ask whether there are \emph{universal} measurements that perform reasonably well for any pair of states and how to quantify their performance.

To address this question, we adopt the framework of Ref.~\cite{matthews_distinguishability_2009}.
Starting point there is a quantum measurement defined via its POVM elements 
--- i.e.\ a family of positive semidefinite operators $\left\{ M_k \right\}_{k=1}^N$ that constitute a partition of the identity
$
	\sum_{k=1}^N M_k = \mathbb{I}.
$
Born's rule asserts that such a POVM maps a state $\rho$
to a discrete probability vector
\begin{equation}
  \mathcal{M} (\rho) = \sum_{k=1}^N |e_k \rangle \mathrm{tr} \left( M_k \rho \right) \in \mathbb{R}^N.
  \label{eq:POVM_outcome_vector}
\end{equation}
In this language, distinguishing $\rho$ from $\sigma$ using the measurement $\mathcal{M}$ reduces to the task of distinguishing between the distributions $\mathcal{M}(\rho)$ and  $\mathcal{M}(\sigma)$. 
The resulting bias will be
\begin{eqnarray*}
	\frac 12 \left\| \mathcal{M}\big( \tau \rho - (1-\tau) \sigma \big) \right\|_{\ell_1} 
	&\leq&
	\frac 12 \left\| \tau \rho - (1-\tau) \sigma  \right\|_{1} ,
\end{eqnarray*}
where the inequality follows from Helstrom's Theorem.
The worst-case ratio between the two sides of the inequality is quantified by the \emph{POVM norm constant}
\begin{equation}
	\lambda_{\mathcal M} = \inf_{X} \frac{\|\mathcal{M}(X)\|_{\ell_1}}{\|X\|_1}, \label{eq:relation_constant}
\end{equation}
where the infimum is over hermitian matrices $X$. 
If additional information is available -- e.g.\ that $\rho, \sigma$ are of high purity or that they occur with equal probability -- it makes sense to define a restricted norm constant by taking the infimum only over differences of weighted states with the given properties. 
In any case, a large value of $\lambda_{\mathcal{M}}$ means that $\mathcal{M}$ preserves distinguishability well.

A conceptually simple measurement that achieves optimal norm constants
\cite[Theorem 10]{matthews_distinguishability_2009}
is the \emph{uniform POVM} $\mathcal{M}_{\textrm{unif}}$.
It maps states of a $d$-dimensional system to probability distributions on the unit sphere of $\mathbb{C}^d$:
\begin{equation*}
	\mathcal{M}_{\mathrm{unif}}(\rho)(\psi)	
	= C_d\, \tr \left( \rho |\psi\rangle \! \langle \psi| \right),
\end{equation*}
where $C_d$ is a suitable normalization constant.
It fulfills
\cite[Theorem 8]{matthews_distinguishability_2009}
\begin{equation*}
  \|\mathcal{M}_{\textrm{unif}}(\rho - \sigma) \|_{\ell_1} \geq \frac{1}{\sqrt{d}} \left( \sqrt{\frac{2}{\pi}}-o(1) \right) \| \rho - \sigma \|_1.
\end{equation*}
When 
$\rho, \sigma$ are assumed to be pure, this improves \cite{matthews_distinguishability_2009} to the dimension-independent statement
\begin{equation}
  \|\mathcal{M}_{\textrm{unif}}(\rho - \sigma) \|_{\ell_1} 
  \geq 
  \frac12 \| \rho - \sigma \|_1.
\label{eq:uniform_pure}
\end{equation}

Refs.~\cite{ambainis_quantum_2007,matthews_distinguishability_2009,lancien_distinguishing_2013} observed that POVMs constructed from complex projective $4$-designs 
(see Sec.~\ref{sec:designs} for definitions)
already essentially match these bounds.
Subsequently, the same has been shown for randomized constructions of POVMs with $\mathcal{O}(d^2)$ different outcomes \cite{aubrun_zonoids_2016}.
However, arguably, none of these POVMs affords a simple structure that would make them easy to analyze further or implement physically.

\subsection{Main Result}

In this paper, we consider POVMs $\mathcal{M}_{C,z}$ whose elements are orbits $\{ U |z\rangle \! \langle z| U^\dagger\}_{U\in \mathrm{C}_n}$ of a \emph{fiducial state} $|z\rangle \! \langle z|$ under the $n$-qubit Clifford group $\mathrm{C}_n$.
The Clifford group (c.f.\ Sec.~\ref{sec:clifford}) 
plays a crucial role in  quantum computation \cite{Gott97the, GottC99, NielC00book, BravK05}, quantum error correction \cite{Gott97the,NielC00book}, randomized
benchmarking \cite{KnilLRB08,MageGE11,WallF14}, and quantum state tomography with compressed sensing \cite{GrosLFB10,Gros11,KimmL15}. 
Importantly, simple and fault-tolerant gate implementations for all elements of the Clifford group are known
\cite{hostens2005stabilizer}.
Our results build on a recent analysis of the representation theory of the 4th tensor power of the Clifford group \cite{other_paper}.

It turns out that the norm constant of $\mathcal{M}_{C,z}$ depend on a simple measure of the ``degree of localization in phase space'' of the vectors in the orbit. 
To state the measure -- identified in \cite{other_paper} -- set $d=2^n$ and let $W_1, \dots, W_{d^2}$ be the set of $n$-qubit Pauli operators (c.f.\ Sec.~\ref{sec:math_prelimi}).
The \emph{characteristic function} of an $n$-qubit quantum state $\rho$ is 
\begin{equation}
  \label{eq:characteristic_function}
  \Xi (\rho)=  \sum_{k=1}^{d^2} \mathrm{tr} \left( W_k \rho \right) |e_k \rangle \in \mathbb{R}^{d^2}. 
\end{equation}
In analogy to the characteristic function that appears e.g.\ in quantum optics, $\Xi$ can be interpreted as a ``phase space representation'' of the operator $\rho$ \cite{Gros06,WallM94book,gross2008quantum}.
Our bounds depend on the quantity
\begin{equation}\label{eq:alpha}
	\alpha (z) 
	= 
	\frac{1}{d^2} \left\| \Xi (|z \rangle \! \langle z|) \right\|_{\ell_4}^4
	=
	\frac{1}{d^2} \sum_{k=1}^{d^2} \big(\tr W_k | z \rangle \! \langle z|\big)^4.
\end{equation}
The value of $\alpha$ is constant along Clifford orbits and bounded between \cite{other_paper}
\begin{equation}
	\frac{2}{d(d+1)}\leq\alpha(z)\leq \frac{1}{d}.
\label{eq:alpha1_bound}
\end{equation}
Smaller values of $\alpha$ turn out to lead to better norm constants. 
At the same time, the number of non-zero coefficients of the characteristic function is lower-bounded by $1/\alpha$, so that Clifford orbits preserve distinguishability well only if their fiducial vector is associated with a ``spread out'' characteristic function.
With these definitions, our main result reads:

\begin{theorem}[Main Theorem]\label{thm:main}
  Fix $d=2^n$, let $|z\rangle$ be a normalized vector in $\mathbb{C}^d$ and let $\mathcal{M}$ be the Clifford POVM generated by $|z\rangle$.
  Then, for all hermitian $X$, it holds that
  \begin{equation*}
	\|\mathcal{M}( X) \|_{\ell_1} \geq 
	\frac{1}{\sqrt{(6 d \alpha (z) r_{\mathrm{eff}}(X) + 10)r_{\mathrm{eff}} (X)}}\, \| X \|_1,
  \end{equation*}
  where $r_{\mathrm{eff}}(X) = \frac{ \| X \|_1^2}{\| X \|_2^2} \leq \mathrm{rank}(X)$ is the \emph{effective rank}.
\end{theorem}

This statement depends both on the effective rank of $X$ and on the choice of the fiducial $|z \rangle \in \mathbb{C}^d$. 
We discuss a number of instances in Sec.~\ref{sec:results}. 
Here, we merely mention the special case of \emph{stabilizer measurements} on pure states.
Stabilizer states \cite{Gott97the,NielC00book}
are extensively studied in quantum information. 
They form an orbit under the Clifford group and thus fall under the scope of Theorem~\ref{thm:main}.

\begin{corollary}[Distinguishing pure states with stabilizer measurements] \label{prop:pure_optimal}
Fix $d=2^n$ and let $\mathcal{M}_{\mathrm{stab}}$ denote the POVM of all $n$-qubit stabilizer states. Then 
\begin{equation}
  \left\| \mathcal{M}_{\mathrm{stab}}(\rho - \sigma) \right\|_{\ell_1}
  \geq 
  \frac{1}{6} \left\| \rho - \sigma \right\|_1	\label{eq:pure_optimal}
\end{equation}
for any pair of \emph{pure states} $\rho,\sigma \in \mathcal{S}_d$.
\end{corollary}

Comparing this to \eqref{eq:uniform_pure} shows that, remarkably, stabilizer measurements perform essentially optimally at distinguishing pure states.

We note that, in fact, \autoref{prop:pure_optimal} is true for any Clifford orbit.

\subsection{Entropic uncertainty relations for stabilizer bases}

The techniques developed in Ref.~\cite{other_paper} and the present paper also allow us to derive new \emph{entropic uncertainty relations} \cite{wehner_entropic_2010,coles_entropic_2015} for stabilizer measurements.

To introduce the concept, let $\mathcal{M}$ be a quantum measurement and $\mathcal{M}(\rho)$ the distribution obtained by applying $\mathcal{M}$ to the quantum state $\rho$ as in \eqref{eq:POVM_outcome_vector}. 
A measure for the ``uncertainty'' that the distribution $\mathcal{M}(\rho)$ leaves about the outcome is given by its Shannon entropy
\begin{equation*}
	H \left( \mathcal{M} | \rho \right) :=  H \left( \mathcal{M}(\rho) \right),
\end{equation*} 
where
\begin{equation}
	H \left( p \right) 
	= - \sum_{k=1}^N p_k \log_2 \left( p_k \right) \in \left[ 0, \log_2 (N) \right]. \label{eq:shannon}
\end{equation}
In the well-known relation attributed to Heisenberg, the uncertainty of a distribution over real numbers is quantified in terms of its variance. 
However, the outcomes of finite POVMs treated here are not usually labeled by real numbers in a natural way. 
Thus, the variance cannot be defined and entropy becomes a more suitable measure.

A typical entropic uncertainty relation captures the incompatibility of several measurements $\mathcal{M}_1,\ldots,\mathcal{M}_L$ by lower-bounding the average entropy associated with the individual outcome probability distributions:
\begin{equation*}
  \frac{1}{L} \sum_{k=1}^L H \left( \mathcal{M}_k | \rho \right) \geq c_{\mathcal{M}_1,\ldots,\mathcal{M}_L} \quad \forall \rho. 
\end{equation*}

For example, a strong entropic uncertainty relation is known 
to hold for measurements 
$\mathcal{B}_1, \dots, \mathcal{B}_{d+1}$
that correspond to a maximal set of \emph{mutually unbiased bases} 
\cite{Ivan81, WootF89, KlapR05M, DurtEBZ10}:
\begin{equation}
\frac{1}{d+1} \sum_{k=1}^{d+1} H \left( \mathcal{B}_k | \rho \right) \geq \log_2 (d+1)-1.	\label{eq:mub_uncertainty}
\end{equation}
Note that this is a strong bound, because each entropic term on the left hand side is bounded from above by $\log_2(d)$ for any choice of the basis measurement.

In Sec.~\ref{sub:entropic_main}, we derive a slightly stronger bound for stabilizer measurements.
Similar to mutually unbiased bases, the set of all $n$-qubit stabilizer states is also a union of 
orthonormal bases.
Denote the associated measurements by $\mathcal{B}_{1},\ldots,\mathcal{B}_{L}$.  
Our main technical result allows us to infer an average entropic uncertainty relation for stabilizer bases that asymptotically outperforms \eqref{eq:mub_uncertainty}:
\begin{equation*}
\frac{1}{L} \sum_{k=1}^{L} H \left( \mathcal{B}_k | \rho \right) \geq \log_2 (d) - c(d),
\end{equation*}
where $\lim_{d \to \infty} c(d) \simeq 0.854$. 
Further results are given in Sec.~\ref{sub:entropic_main}.

\section{Results}

In this section, we briefly introduce  complex projective designs and Clifford orbits, before stating the results of the present paper.

\subsection{Complex Projective Designs}
\label{sec:designs}

We will frequently compare the results we obtain for Clifford orbits to those that are valid for \emph{complex projective designs} \cite{DelsGS77,Hogg82,ReneBSC04, Scot06,AmbaE07}. 
An introduction to the theory and applications of designs is given in the companion paper \cite{other_paper}. 
Here, we merely state the definition.

\begin{definition}[Complex projective design]
A \emph{complex projective $t$-design} is a set of unit vectors $\left\{|x_k\rangle \right\}_{k=1}^N \subset \mathbb{C}^d$ whose outer products obey
\begin{align}
\frac{1}{N} \sum_{k=1}^N |x_k \rangle \! \langle x_k |^{\otimes t} &= \int_{\|v\|_{\ell_2} = 1} \mathrm{d} v |v \rangle \! \langle v|^{\otimes t}\nonumber\\
&=\binom{d+t-1}{t}^{-1}P_{\mathrm{Sym}^t}. \label{eq:tdesign}
\end{align}
Here integration on the r.h.s. is taken with respect to the uniform measure on the complex unit sphere in $\mathbb{C}_d$, and $P_{\mathrm{Sym}^t}$ is the projector onto the totally symmetric subspace of $(\mathbb{C}^d)^{\otimes t}$.
Likewise, we call the set 
\begin{equation}
\mathcal{M}_{\mathrm{tD}} = \left\{ \frac{d}{N} |x_k\rangle \!\langle x_k | \right\}_{k=1}^N
\end{equation}
 a \emph{$t$-design POVM}.
\end{definition}

We find it fruitful to think of Eq.~(\ref{eq:tdesign}) as saying that drawing vectors uniformly from a complex projective $t$-design reproduces the first $2\cdot t$ moments of Haar-random vectors.

It is known \cite{ambainis_quantum_2007,matthews_distinguishability_2009,lancien_distinguishing_2013}
that 4-design POVMs perform essentially optimally at the task of distinguishing quantum states.
More precisely,
\autoref{thm:4design} below---which is a slight improvement over existing results in \cite{lancien_distinguishing_2013}---implies
\begin{equation}
\| \mathcal{M}_{\mathrm{4D}}(X) \|_{\ell_1} > \frac{0.32}{\sqrt{ \mathrm{rank}(X)}} \| X \|_1 \quad \forall X \in \mathcal{H}_d. \label{eq:4design}
\end{equation}

In stark contrast to this, 2-design POVMs perform very poorly at distinguishing pure quantum states \cite[Theorem 12]{matthews_distinguishability_2009}:
\begin{equation}
\| X \|_{\mathcal{M}_{\textrm{2D}}} \geq \frac{1}{2(d+1)} \| X \|_1 \quad \forall X \in H_d. \label{eq:2design}
\end{equation}
The pre-factor $\frac{1}{d+1}$ in this relation is in general unavoidable, even if $X = \rho - \sigma$ is a difference of pure states \cite[Section 2.C]{matthews_distinguishability_2009}.

\subsection{Clifford Orbits}
\label{sec:clifford}

The Clifford group \cite{Gott97the, GottC99, NielC00book, BravK05} can be defined as the normalizer of the group generated by the Pauli operators.
Alternatively, for dimensions $d=2^n$ that are a power of two, the Clifford group is the group generated by Pauli operators, the Hadamard gate, phase gate, and the controlled-NOT gate.
Again, a more complete treatment is given in the companion paper \cite{other_paper}.

The multi-qubit Clifford group has a very rich structure. 
Relevant for our result is that it forms a \emph{unitary $3$-design} 
\cite{zhu_multiqubit_2015,webb_clifford_2015}. 
Unitary $t$-designs are a generalization of  complex projective $t$-design  to unitary matrices \cite{dankert_exact_2009,gross_evenly_2007}.
They have the particular property that every orbit of a unitary $t$-design forms a complex projective $t$-design.
This in turn implies that every multi-qubit Clifford POVM is also a 3-design POVM. 
For the most prominent orbit -- the set of all stabilizer states -- the 3-design property has been established independently \cite{kueng_qubit_2015}.

Unfortunately, POVMs derived from designs of degree $t=3$ do not achieve optimal norm constants and it has been shown that neither is the Clifford group a unitary 4-design \cite{zhu_multiqubit_2015,webb_clifford_2015}, nor does the weaker statement hold that stabilizer states form a complex projective 4-design \cite{kueng_qubit_2015}.
However, in the companion paper \cite{other_paper}, alternative methods for analyzing the $8$th moments of Clifford orbits have been established.
These results form the basis for the discussion of Clifford POVMs below.

\begin{definition}[Clifford POVM] \label{def:clifford}
Set $d=2^n$ and fix $|z\rangle \in \mathbb{C}^d$ with unit length. Let $\left\{ C_k |z \rangle\langle z|C_k^\dag :\; C_k \in \mathrm{C}_n \right\}$ denote the orbit of $|z\rangle\langle z|$ under the Clifford group and $N$ its cardinality.
We then define the associated \emph{Clifford POVM} to be
\begin{equation*}
\mathcal{M}_{C,z} = \left\{ \frac{d}{N} C_k |z \rangle \! \langle z| C_k^\dagger: \; C_k \in \mathrm{C}_n \right\}.
\end{equation*}
\end{definition}

\subsection{Technical results} 
\label{sec:results}

Recall the statement of \autoref{thm:main}:
\begin{equation}
\| \mathcal{M}_{C,z} (X) \|_{\ell_1} \geq \frac{ \| X \|_1}{\sqrt{(6 d \alpha (z) r_{\mathrm{eff}}(X) + 10)r_{\mathrm{eff}} (X)}} \label{eq:main_clifford}
\end{equation}
for any$X\in H_d$.
This statement depends on the choice of fiducial via $\alpha (z)$ introduced in \eqref{eq:alpha}.
It is worthwhile to point out that, unlike its counterparts  for 4- and 2-design POVMs, Formula~\eqref{eq:main_clifford} is sensitive to the effective rank of the matrix $X$ considered:
\begin{equation}
\|  \mathcal{M}_{C,z} (X)  \|_{\ell_1} \geq \frac{ \| X \|_1}{4 \sqrt{r_{\mathrm{eff}}(X)}}, \label{eq:bound_good}
\end{equation}
provided that $r_{\mathrm{eff}}(X) \leq \frac{1}{d \alpha (z)}$ Otherwise:
\begin{equation}
\|  \mathcal{M}_{C,z} (X) \|_{\ell_1} \geq \frac{\| X \|_1}{4 r_{\mathrm{eff}}(X) \sqrt{d\alpha (z)}}. \label{eq:bound_bad}
\end{equation}
Thus, if $r_{\mathrm{eff}}(X)$ is below a certain threshold (which depends on the choice of the fiducial), the favorable bound \eqref{eq:bound_good} applies. Such a situation is comparable to the 4-design case.
However, above this threshold one needs to resort to the much weaker bound \eqref{eq:bound_bad}.  Depending on the choice of fiducial, its scaling may be comparable to the 2-design case, once $r_{\mathrm{eff}}(X)$ approaches $d$. Fortunately, in Ref.~\cite{other_paper} we have shown that the value of $\alpha(z)$ for a typical fiducial $|z\rangle$ is very close to the value required for a 4-design [see \eqref{eq:TypicalFiducial}], so typical Clifford POVMs perform almost as well as 4-design POVMs.

The following converse statement shows that the aforementioned  behavior is essentially unavoidable for certain Clifford orbits.

\begin{theorem} \label{thm:converse}
Fix $d=2^n$, let $\mathcal{M}_{C,z}$ denote a Clifford POVM with fiducial $|z\rangle \in \mathbb{C}^d$ and fix $W \in \mathcal{H}_d$ to be any Pauli matrix, $W \neq \mathbb{I}$. 
Then
\begin{equation}\label{eq:converse}
\|  \mathcal{M}_{C,z} (W)  \|_{\ell_1} =\frac{ \|\Xi (|z \rangle \! \langle z| ) \|_{\ell_1}-1}{(d+1)(d-1)}\| W \|_1. 
\end{equation}
\end{theorem}
The coefficient in the theorem satisfies
\begin{equation}\label{eq:converse2}
\frac{1}{d+1}\leq \frac{\|\Xi (|z \rangle \! \langle z| ) \|_{\ell_1}-1}{(d+1)(d-1)}\leq \frac{1}{\sqrt{d+1}}
\end{equation}
which follows from the properties of 
the characteristic function for a pure state. The lower bound is saturated if and only $z \in \mathbb{C}^d$ is a stabilizer state, and the upper bound is saturated iff 
\begin{equation*}
|\langle z|W_k|z\rangle|=\frac{1}{\sqrt{d+1}},\quad  \forall 2\leq k\leq d^2,
\end{equation*}
in which case the orbit of $z$ under the action of the Pauli group forms a symmetric informationally complete POVM \cite{renes_symmetric_2004}. 
Moreover, the pre-factor in \eqref{eq:converse} may be related to $\alpha(z)$---the main figure of merit in \autoref{thm:main}. We provide such a relation in Eq.~\eqref{eq:converse3} below.

We now move on to discussing the implications of our findings for four different Clifford orbits.
\begin{enumerate}
\item[(i)] \emph{Stabilizer states:} multi-qubit stabilizer states form a particular Clifford orbit $\mathcal{M}_{\mathrm{stab}}$ with $N= 2^n \prod_{j=1}^n \left( 2^j + 1 \right)$ elements. The characteristic function of any stabilizer state 
has precisely $d$ non-vanishing components with constant modulus 1---see Sec.~\ref{sec:fiducials} below. 
This in turn implies $d \alpha (z)=1$ for any stabilizer state fiducial $|z\rangle \in \mathbb{C}^d$. Consequently, the favorable bound in \autoref{thm:main} is only valid for rank-one matrices $X$, where $\sqrt{r_{\mathrm{eff}}(X)}$ and $r_{\mathrm{eff}}(X)$ coincide. In turn we need to conclude
\begin{equation}
\| \mathcal{M}_{\mathrm{stab}}(X) \|_{\ell_1} \geq \frac{1}{4 r_{\mathrm{eff}}(X)} \| X \|_1, \label{eq:main_stab}
\end{equation}
for any $X \in \mathcal{H}_d$.
This is a worst case behavior for any Clifford orbit. However, \autoref{thm:converse} assures that such a scaling is unavoidable: the characteristic function of stabilizer states obeys $\|\Xi (|z \rangle\! \langle z|) \|_{\ell_1} = d$ and inserting this into \eqref{eq:converse} reveals
\begin{equation}
\| \mathcal{M}_{\mathrm{stab}}(W) \|_{\ell_1} = \frac{d}{d+1} \frac{ \| W \|_1}{r_{\mathrm{eff}} (W)} \label{eq:stab_converse}
\end{equation}
for any Pauli matrix $W \neq \mathbb{I}$. This equation implies that \eqref{eq:main_stab} is actually tight up to a multiplicative constant.

\item[(ii)] \emph{Magic state fiducial:} Let $|z \rangle \! \langle z| = \rho_{\textrm{magic}}^{\otimes n}$ be the $n$-fold tensor product of  the single qubit \emph{magic state'}
\begin{equation*}
\rho_{\textrm{magic}} = \frac{1}{2} \left( \mathbb{I} + \frac{1}{\sqrt{3}} \left( \sigma_1 + \sigma_2 + \sigma_3 \right) \right) \in \mathcal{S}_2,
\end{equation*} 
where $\sigma_1,\sigma_2,\sigma_3 \in \mathcal{H}_2$ denote the single-qubit Pauli matrices.
Such a fiducial obeys $d \alpha (z) =(\frac{2}{3})^n< \frac{1}{\sqrt{d}}$ (see Eq.~\eqref{eq:magic_4norm} below). This is considerably smaller than the analogous quantity for stabilizer states.
In turn, \autoref{thm:main} implies that Clifford POVMs 
with a magic state fiducial obey
\begin{equation}
\| \mathcal{M}_{C,\textrm{magic}}(X) \|_{\ell_1} \geq \frac{1}{4 \sqrt{ r_{\mathrm{eff}}(X)}} \| X \|_1 \label{eq:magic_good}
\end{equation} 
for any $X \in \mathcal{H}_d$ with $r_{\mathrm{eff}}(X) \leq (\frac{3}{2})^n$. 
For matrices $X$ whose effective rank exceeds $(\frac{3}{2})^n$, \autoref{thm:main} still implies
\begin{equation}
\| \mathcal{M}_{C,\textrm{magic}}(X) \|_{\ell_1} \geq  \frac{(\frac{3}{2})^{n/2} \| X \|_1}{4 r_{\mathrm{eff}}(X)}> \frac{ d^{0.29} \| X \|_1}{4 r_{\mathrm{eff}}(X)} \label{eq:magic}
\end{equation}
which outperforms the analogous bound for stabilizer states by a factor of $d^{0.29}$. Conversely, \autoref{thm:converse} requires
\begin{equation}
\| \mathcal{M}_{C,\textrm{magic}}(W) \|_{\ell_1}  \leq \frac{d^{0.45} \| W \|_1}{r_{\mathrm{eff}}(W)}, \label{eq:magic_converse}
\end{equation}
because $\|\Xi (|z \rangle \! \langle z|)\|_{\ell_1} =(1+\sqrt{3})^n \leq d^{1.45}$ (see Eq.~\eqref{eq:magic_1norm} below).
Unlike before, this bound is too weak to ensure tightness of \eqref{eq:magic}.
However, asymptotically it does rule out the possibility of an optimal scaling for this type of Clifford orbits.

\item[(iii)] \emph{4-design fiducial:} As pointed out in \cite{other_paper}, particular choices of fiducials $|z \rangle \in \mathbb{C}^d$ result in Clifford orbits that actually form a complex projective 4-design. 
The necessary and sufficient requirement for such fiducials is $\alpha (z)= \frac{4}{(d+3)d}$. 
According to \autoref{thm:4design} below,
\begin{equation*}
\| \mathcal{M}_{C,\textrm{4D}}(X) \|_{\ell_1} \geq \frac{0.32}{ \sqrt{ \mathrm{rank}(X)}} \| X \|_1 \quad \forall X \in \mathcal{H}_d.
\end{equation*}
 This bound is optimal up to  a small multiplicative constant. Combining \autoref{thm:converse} with \eqref{eq:converse2} demands
\begin{align*}
\| \mathcal{M}_{C,\textrm{4D}}(W) \|_{\ell_1} \leq  \frac{ \| W\|_1}{\sqrt{d+1}} <\frac{ \| W\|_1   }{\sqrt{\mathrm{rank}(W)} }
\end{align*}
for any Pauli matrix $W$ that is not proportional to the identity. 

\item[(iv)] \emph{Typical fiducial:} According to \cite{other_paper}, if $|z\rangle$ is distributed uniformly on the complex unit sphere in  $\mathbb{C}^d$, then the following inequality  
\begin{equation}\label{eq:TypicalFiducial}
\alpha (z) \leq \frac{6}{(d+3)d}
\end{equation} 
is satisfied with very high probability. 
Such orbits behave almost like 4-designs and \autoref{thm:main_generic} below implies
\begin{equation*}
\| \mathcal{M}_{C,z}(X) \|_{\ell_1} \geq \frac{ \| X \|_1}{\sqrt{22 r_{\mathrm{eff}}(X)}} \quad \forall X \in \mathcal{H}_d.
\end{equation*}

\end{enumerate}

\subsection{Implications for distinguishing quantum states} \label{sub:main_distinguishability}

Let us now turn back our attention to the task of distinguishing different quantum states in the single shot scenario. 
Matthews \emph{et al.} introduced the POVM norm constant $\lambda_{\mathcal{M}}$ \eqref{eq:relation_constant} to compare the performance of a fixed POVM $\mathcal{M}$ directly to Helstrom's optimal strategy. Without putting further restrictions on the states $\rho,\sigma \in \mathcal{S}_d$ to be distinguished, \autoref{thm:main} only allows us to infer
\begin{equation}\label{eq:lambda_bound}
\lambda_{\mathcal{M}_{C,z}} \geq \frac{1}{\sqrt{d(6d^2 \alpha (z) +10)}}.
\end{equation}
for Clifford POVMs with fiducial $|z\rangle \in \mathbb{C}^d$. For the particular case of multi-qubit stabilizer states, we have
\begin{equation}\label{eq:stab_bound}
\frac{1}{\sqrt{6}d} \leq \lambda_{\mathcal{M}_{\textrm{stab}}} \leq \frac{1}{d+1}.
\end{equation}
Here the lower bound is derived in Sec.~\ref{sec:Proofs};
the upper bound follows from \eqref{eq:stab_converse}\footnote{Every Pauli matrix $W$ has vanishing trace and is therefore proportional to a particular difference $\tau \rho - (1-\tau) \sigma$ of quantum states $\rho,\sigma \in \mathcal{S}_d$ with $\tau = \frac{1}{2}$.}.
This result shows that the constant $\lambda_{\mathcal{M}_{\textrm{stab}}}$ scales like $\lambda_{\mathcal{M}_{\textrm{2D}}}$ from \eqref{eq:2design}---despite the fact that multi-qubit stabilizers form  a 3-design.

For Clifford orbits with a magic state fiducial we obtain
\begin{equation*}
\frac{1}{4 d^{0.71}} \leq \lambda_{\mathcal{M}_{C,\textrm{magic}}} \leq \frac{1}{d^{0.55}}.
\end{equation*}
 Qualitatively, this bound assures that the capacity of such POVMs to distinguish quantum states is ``half way'' between the existing 2-design ($\lambda_{\textrm{2D}} \geq \frac{1}{2(d+1)}$ ) and 4-design guarantees ($\lambda_{\textrm{4d}} \geq \frac{0.32}{\sqrt{d}}$). 
Naively, one may expect precisely such a behavior for 3-designs.

Finally, Clifford POVMs with  typical random fiducials perform considerably better. Indeed, \autoref{thm:main_generic} assures
\begin{equation*}
\lambda_{\mathcal{M}_{C,z}} \geq \frac{1}{\sqrt{22d}},
\end{equation*}
which---up to a multiplicative constant---reproduces the close-to-optimal 4-design case. Clearly, this in particular extends to Clifford POVMs with a 4-design fiducial.

We emphasize that the constant $\lambda_{\mathcal{M}}$ is a worst case promise for correctly distinguishing \emph{any} pair of states $\rho, \sigma \in \mathcal{S}_d$. 
This may be too pessimistic for more concrete scenarios where additional structure is present. 
One model assumption, which is often met in practice, is approximate purity. In the extreme case, where both $\rho$ and $\sigma$ are assumed to be pure, \autoref{thm:main} assures $\lambda_{\mathcal{M}_{C,z}}|_{\rho,\sigma\textrm{ pure}} \geq \frac{1}{\sqrt{44}}$
for any Clifford orbit, including stabilizer states. A slightly better bound was presented in \autoref{prop:pure_optimal},
\begin{equation*}
\lambda_{\mathcal{M}_{C,z}}|_{\rho,\sigma\textrm{ pure}} \geq \frac{1}{6}.
\end{equation*}
 Up to a multiplicative constant, this reproduces the 4-design behavior. 
It is worthwhile to point out that 2-design POVMs do not allow for exploiting purity at all \cite[Section 2.C]{matthews_distinguishability_2009}.

Similar conclusions may be drawn if we relax the model assumption of purity to low effective rank $r \ll d$. 
As the rank constraint $r$ increases, the bounds on $\lambda_{\mathcal{M}_{C,z}}|_{\rho,\sigma \textrm{ rank $r$}}$ become gradually weaker until they approach \eqref{eq:lambda_bound} for $r_{\mathrm{eff}}=d$.

Finally, we point out that the notion of effective rank is useful for several concrete applications.
Consider for instance the task of deciding whether a pure state $\phi = | \phi \rangle \! \langle \phi|$, or the maximally mixed state $\frac{1}{d} \mathbb{I}$ was prepared. \autoref{lem:effective_rank} below assures that $X=\frac{1}{2} \phi - \frac{1}{2d}\mathbb{I}$ has effective rank less than 4 and consequently \eqref{eq:main_stab} implies
\begin{equation*}
\left\|\mathcal{M}_{C,z} \left(  \phi - \frac{1}{d} \mathbb{I} \right) \right\|_{\ell_1} \geq \frac{1}{16} \left\|   \phi - \frac{1}{d} \mathbb{I} \right\|_1
\end{equation*}
for any Clifford orbit. 
This implies that the optimal bias achievable with such a POVM measurement is directly comparable to Helstrom's optimal one. We will use such generalizations for deriving the entropic certainty relations presented in the next section.

\subsection{Entropic uncertainty and certainty relations for stabilizer bases} \label{sub:entropic_main}

Stabilizer states form the  most structured Clifford orbit. Similar to a maximal set of mutually unbiased bases, multi-qubit stabilizer states form a union of $\frac{N}{d} = \prod_{j=1}^n \left( 2^j+1 \right)$ different orthonormal bases $\mathcal{B}_1,\ldots,\mathcal{B}_{N/d}$.
These stabilizer bases obey the same entropic uncertainty relation as mutually unbiased bases do:
\begin{equation}
\frac{d}{N} \sum_{k=1}^{N/d} H \left( \mathcal{B}_k | \rho \right) \geq \log_2 (d+1) -1 \label{eq:stab_uncertainty}.
\end{equation}
As pointed out in \cite{wehner_entropic_2010} this strong entropic uncertainty relation may be derived from the fact that both stabilizer states and mutually unbiased bases form complex projective 2-designs. 
We present a derivation of this statement in Sec.~\ref{sec:entropic_proofs}.
However, this proof technique does not allow for establishing stronger uncertainty relations for designs of higher order.

We partially overcome this lack of proof techniques by formulating a linear programming problem whose solution provides a lower bound on the average entropy of measurements in stabilizer bases. 
Unlike the derivation of \eqref{eq:stab_uncertainty}, the moment constraints of higher $t$-designs do feature as constraints in said linear program. 
This allows us to advantageously take into account additional information about the third and fourth moments of multi-qubit stabilizer states. 
We obtain the following uncertainty relation for dimensions $d=2^n$:
\begin{equation}
\frac{d}{N} \sum_{k=1}^{N/d} H \left( \mathcal{B}_k | \rho \right) \geq \log_2 (d) - c (d), \label{eq:entropic_uncertainty}
\end{equation}
with $\lim_{n\to \infty} c \left( d \right) \simeq 0.854 <1$. Similar statements may be formulated for other Clifford orbits. These findings are detailed in \autoref{fig:uncertainty} and provide an affirmative answer to an open problem formulated by Wehner and Winter in \cite{wehner_entropic_2010}:
is it possible to take advantage of  higher design structures $t \geq 3$ when formulating entropic uncertainty relations?

\begin{figure}
\includegraphics[width=\linewidth]{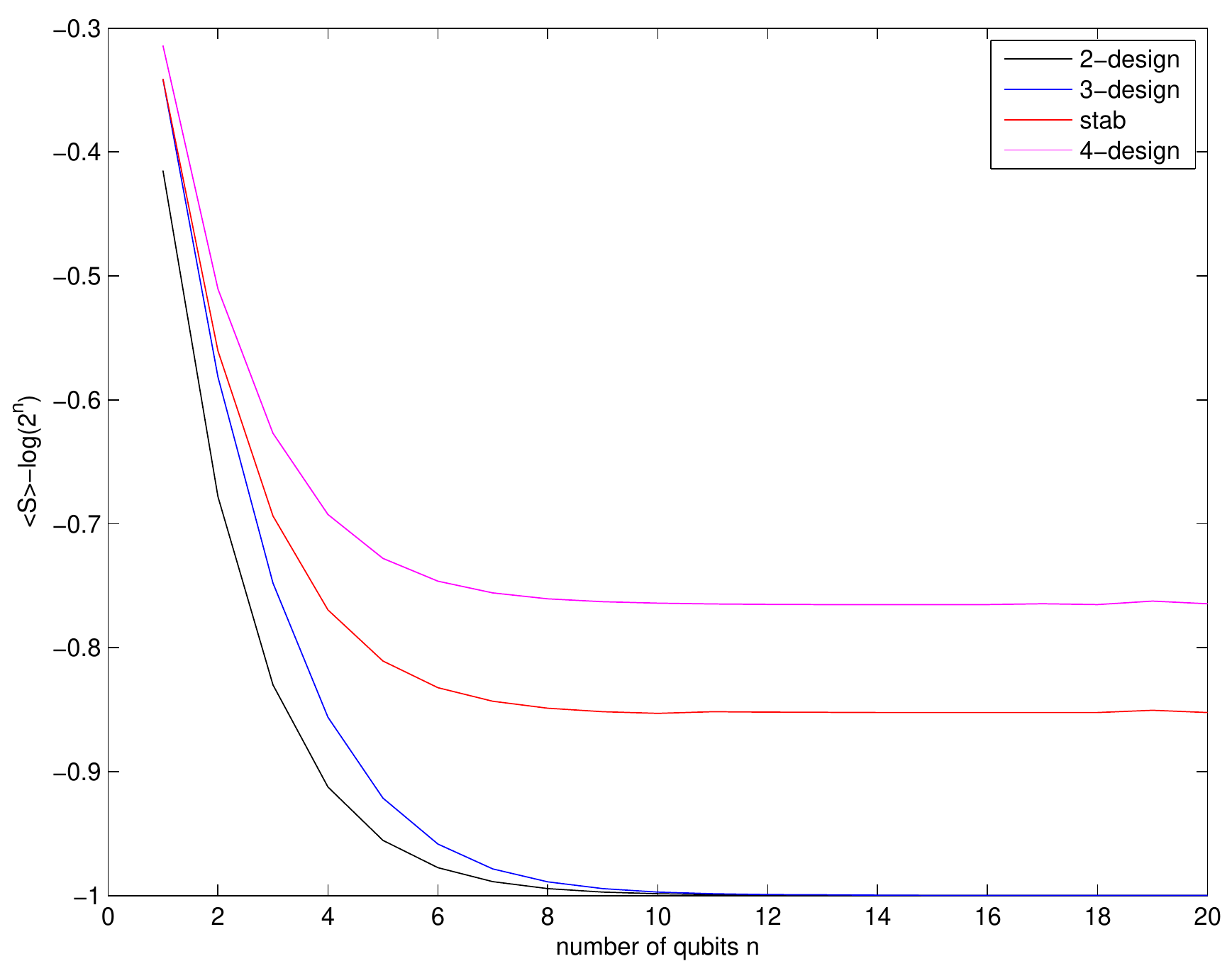}
\caption{Entropic uncertainty relations for stabilizer bases. The red line indicates the value of $-c(d)$ in Eq.~(\ref{eq:entropic_uncertainty}) for dimension $d=2^n$ ranging from $n=1$ to $n=20$. 
We obtain these values via the linear programming approach detailed in Sec.~\ref{sec:entropic_proofs}.
As a reference, corresponding values for 2-designs (black), 3-designs (blue) and 4-designs (magenta) are plotted as well. 
Note that the structure of a 3-design alone does not allow for an asymptotic improvement over \eqref{eq:stab_uncertainty}, as the blue line converges to $-1$ as $n$ increases.
However, taking into account additional information about fourth moments of stabilizer states does lead to consistently better results: $c(d) \simeq 0.854$.
}
\label{fig:uncertainty}
\end{figure}

Following Matthews \emph{et al.} \cite{matthews_distinguishability_2009}, we may also employ knowledge about the third and fourth moments of stabilizer states to obtain entropic bounds in the converse direction.
Introduced by Sanchez-Ruiz \cite{sanchez_improved_1995}, these \emph{certainty relations} provide a lower bound on the information that is accessible via such measurements. 
We refer to \cite{matthews_distinguishability_2009} for further clarification of the terminology used here.
Let us consider an isotropic ensemble $\sum_{x} p_x \rho_x = \frac{1}{d} \mathbb{I} \in \mathcal{S}_d$ of quantum states and a POVM measurement $\mathcal{M}_{C,z}$ that corresponds to an arbitrary Clifford orbit, including stabilizer states.
Then the \emph{Shannon mutual information} between the preparation variable $X$ and the measurement outcome $\mathcal{M}_{C,d}$ obeys
\begin{equation*}
I (X:\mathcal{M}_{C,d}) \geq  \frac{1}{128 \log (2)} \left( \frac{d-1}{d} \right)^2
\end{equation*} 
for any fiducial $z \in \mathbb{C}^d$. Regardless of the particular choice of Clifford orbit, this is a small, but finite, constant that lower bounds the ``accessible information''. For particular Clifford orbits, this bound may be improved further. Clifford POVMs with a 4-design fiducial, for instance, 
admit $I(X:\mathcal{M}_{C,d}) = I(X: \mathcal{M}_{\textrm{4D}}) \geq \frac{1}{18 \log (2)} \left( \frac{d-1}{d} \right)^2$ \cite{matthews_distinguishability_2009}.

Moreover, multi-qubit Clifford  POVMs obey
\begin{equation}
H \left( \mathcal{M}_{C,z} | \phi  \right) \leq \log_2 (N) - \frac{1}{128 \log (2)} \left( \frac{d-1}{d} \right)^2 \label{eq:certainty2}
\end{equation}
for any pure state $\phi = |\phi \rangle \! \langle \phi| \in \mathcal{H}_d$. Again, it is possible to further improve the constant for Clifford orbits with particular structure.

It is worthwhile to compare this relation to a similar one derived by Matthews \emph{et al.} \cite{matthews_distinguishability_2009} for 2-design POVMs $\mathcal{M}_{\mathrm{2D}}: \mathcal{H}_d \to \mathbb{R}^N$:
\begin{equation*}
H \left( \mathcal{M}_{\mathrm{2D}} |\phi  \right) \leq \log (N)- \frac{1}{6 \log (2)}\frac{1}{(d+1)^2}.
\end{equation*}
Note that asymptotically ($d \to \infty$) the Clifford certainty relation \eqref{eq:certainty2} is much tighter than this 2-design analogue. 

These findings highlight that Clifford  POVMs, including 
stabilizer measurements in particular, obey strong uncertainty and certainty relations. 
This agrees with previous studies about entropic uncertainty relations for fixed pairs of stabilizer basis measurements, see e.g.\ \cite{niekamp_entropic_2012}.

\section{\label{sec:Proofs}Proofs of the main technical results}

\subsection{Mathematical preliminaries}
\label{sec:math_prelimi}

Throughout this work we will exclusively consider dimensions $d=2^n$ that are a power of 2. Let $W_1,\ldots,W_{d^2} \in \mathcal{H}_d$ denote the $d^2$ Pauli operators and $\Xi (\cdot)$ the associated characteristic function introduced in \eqref{eq:characteristic_function}. Also, note that $d=2^n$ assures that every $W_k$ is actually a tensor product $W_k = \otimes 
\sigma_{k_1} \otimes \cdots \otimes \sigma_{k_n}$ of single qubit Pauli matrices $\sigma_0,\sigma_1,\sigma_2,\sigma_3 \in \mathcal{H}_2$.  

We endow the vector spaces $\mathbb{C}^{d'}$ and $\mathbb{R}^{d'}$ with the usual $\ell_p$-norms.
On the level of hermitian matrices $X\in \mathcal{H}_d$, let $|X| = \sqrt{ X X^\dagger}$ denote the matrix absolute value. 
We then define the Schatten-$p$-norms to be $\| X \|_p = \left( \mathrm{tr} \left( |X|^p \right) \right)^{1/p}$. These are related via $\| X \|_q \leq \| X \|_p$ for all $X \in \mathcal{H}_d$ and $p \leq q$. Moreover, the trace norm ($p=1$) and the Hilbert-Schmidt norm ($p=2$) obey the following converse relation: $\| X \|_1 \leq \sqrt{ \mathrm{rank}(X)} \| X \|_2$ $\forall X \in \mathcal{H}_d$. 

The main technical prerequisite for \autoref{thm:main} is the following statement.

\begin{theorem}[\cite{other_paper}]  \label{thm:main_technical}
Fix $d=2^n$ and let $\mathcal{M}_{C,z}= \left\{ |x_k \rangle\right\}_{k=1}^N \subseteq \mathbb{C}^d$ be a Clifford orbit with fiducial $|z\rangle \in \mathbb{C}^d$ and $N$ elements. Then 
\begin{equation*}
\frac{1}{N} \sum_{k=1}^N \left( | x_k \rangle \! \langle x_k |\right)^{\otimes 4} = d \binom{d+2}{3}^{-1} \left( \alpha (z) P_1 + \beta (z) P_2 \right),
\end{equation*}
where $P_1, P_2$ are orthogonal projectors that sum up to $P_{\mathrm{Sym}^4}$, $\alpha (z)$ was defined in \eqref{eq:alpha} and
$
\beta (z) = \frac{4(1- \alpha (z))}{(d+4)(d-1)}.
$
Defining $Q= \frac{1}{d^2}\sum_{k=1}^{d^2} W_k^{\otimes 4}$ allows for characterizing the projectors explicitly by
\begin{equation*}
P_1 = P_{\mathrm{Sym}^4} Q \quad\textrm{and} \quad P_2 = P_{\mathrm{Sym}^4} \left( \mathbb{I} - Q \right).
\end{equation*}
\end{theorem}
Note that $Q$ is a projector that commutes with $P_{\mathrm{Sym}^4}$. In addition,
the right hand side of the statement in \autoref{thm:main_technical} may be rewritten as
\begin{equation}
d \binom{d+2}{3}^{-1}
\left( \left( \alpha (z) - \beta (z) \right) P_1 + \beta (z) P_{\mathrm{Sym}^4} \right). \label{eq:main_technical}
\end{equation}
According to \eqref{eq:alpha1_bound}, the difference between these coefficients obeys
\begin{equation}
 -\frac{2}{d(d+1)} \leq \alpha (z)-\beta (z)\leq \frac{1}{d+4} \label{eq:alphabound2}
\end{equation}

It is insightful to compare this statement to the defining property \eqref{eq:tdesign} of a complex projective 4-design.
\begin{align}
\frac{1}{N} \sum_{k=1}^n (|x_k \rangle \! \langle x_k |)^{\otimes 4}
=& \binom{d+3}{4}^{-1} P_{\mathrm{Sym}^4}.  \label{eq:4design_loc}
\end{align}
From such a comparison it becomes apparent that Clifford orbit fiducials $|z \rangle \in \mathbb{C}^d$ result in a complex projective 4-design, precisely if $\alpha (z)=\frac{4}{d(d+3)}$. Indeed, such a choice assures $\alpha (z) = \beta (z) = \frac{4}{d(d+3)}$ for the constants occurring in \autoref{thm:main_technical}
which in turn implies the defining property \eqref{eq:4design_loc} of a 4-design.

However, \autoref{thm:main_technical} also implies that Clifford orbits in general do not have this very particular behavior and consequently fall short of being complex projective 4-designs.
Fortunately, the deviation from this ideal behavior is benign: the fourth moment average decomposes into exactly two projectors $P_1,P_2$ instead of a single one, namely $P_{\mathrm{Sym}^4}$. 
As we shall see, this deviation is mild enough to adapt the proof technique from the 4-design statement by Ambainis and Emerson \cite{ambainis_quantum_2007} (see also \cite[Section 2.B]{matthews_distinguishability_2009} and \cite{lancien_distinguishing_2013}) to Clifford orbits.

\subsection{A novel bound for 4-design POVMs}

In this section, we present a slight improvement over previous results regarding distinguishability of quantum states via 4-design POVMs. Its proof outline will serve as a guideline for the derivation of our main technical result: \autoref{thm:main}.

\begin{theorem}[Performance of 4-designs] \label{thm:4design}
Let $\mathcal{M}_{\mathrm{4D}}$ be a 4-design POVM. Then
\begin{equation}
\| \mathcal{M}_{\mathrm{4D}}(X) \|_{\ell_1} > \frac{0.32}{\sqrt{ \mathrm{rank}(X)}} \| X \|_1 \quad \forall X \in \mathcal{H}_d. \label{eq:4design_bound}
\end{equation}
This in particular implies that the distinguishability constant \eqref{eq:relation_constant} obeys 
$
\lambda_{\mathcal{M}_{\textrm{4D}}} > \frac{0.32}{\sqrt{ d}}$. If $X$ has rank two, then the constant $0.32$ may be further improved to $\frac{1}{\sqrt{6.06}}> 0.4$.
\end{theorem}

The original statements in \cite{ambainis_quantum_2007,matthews_distinguishability_2009} require $X$ to be traceless, while \cite{lancien_distinguishing_2013} affords a slightly smaller constant constant of $\frac{1}{\sqrt{18}}$.
Also, $X$ having rank two encompasses the case of distinguishing two pure quantum states. Our statement provides a tighter constant for this particularly relevant special case.

At the heart of the proof of \autoref{thm:4design} (see e.g.\ \cite{ambainis_quantum_2007,matthews_distinguishability_2009}) is the following moment inequality by Berger \cite{berger_fourth_1997}:
\begin{equation}
\mathbb{E} \left[ \left| S \right| \right]
\geq \sqrt{ \frac{ \mathbb{E} \left[ S^2 \right]^3}{\mathbb{E} \left[ S^4 \right]}}. \label{eq:berger}
\end{equation}
It is valid for any real valued random variable $S$. 

Now, let $\mathcal{M}_{\textrm{4D}}=\left\{ \frac{d}{N} |x_k \rangle\! \langle x_k | \right\}_{k=1}^N$ be a 4-design POVM, fix $X \in \mathcal{H}_d$ arbitrary and define the $N$-variate random variable
\begin{equation}
S_X = \langle x_k | X |x_k \rangle \quad \textrm{with probability} \quad \frac{1}{N}. \label{eq:S}
\end{equation}
Accordingly,
\begin{align}
\| \mathcal{M}_{\textrm{4d}}(X) \|_{\ell_1}
=&
\frac{d}{N} \sum_{k=1}^N \left| \langle x_k | X |x_k \rangle \right|=  d \mathbb{E} \left[ \left| S_X \right| \right] \nonumber \\
\geq &   d \sqrt{ \frac{ \mathbb{E} \left[ S_X^2 \right]^3}{\mathbb{E} \left[ S_X^4 \right]}}. \label{eq:4design_aux1}
\end{align}
So in order to establish \autoref{thm:4design}, it suffices to bound the moments $\mathbb{E} \left[ S_X^2 \right]$, as well as $\mathbb{E} \left[ S_X^4\right]$ appropriately.
Since any complex projective 4-design in particular also constitutes a 2-design, the first quantity amounts to
\begin{align}
\mathbb{E} \left[ S_X^2 \right]
=& \frac{1}{N} \sum_{k=1}^N \mathrm{tr} \left( |x_k \rangle \! \langle x_k | X\right)^2 \nonumber \\
=& \mathrm{tr} \left( \frac{1}{N} \sum_{k=1}^N \left( |x_k \rangle \! \langle x_k | \right)^{\otimes 2} X^{\otimes 2} \right) \nonumber \\
=& \binom{d+1}{2}^{-1} \mathrm{tr} \left( P_{\mathrm{Sym}^2} X^{\otimes 2} \right) \nonumber \\
=& \frac{ \mathrm{tr} \left( X^2 \right) + \mathrm{tr}(X)^2}{(d+1)d}, \label{eq:2moment_bound}
\end{align}
where the last equation follows from $P_{\mathrm{Sym}^2} = \frac{1}{2} \left( \mathbb{I} + \mathbb{F} \right)$ with $\mathbb{F}$ denoting the Flip-operator on a bi-partite system (see e.g. \cite[Lemma 6]{gross_partial_2015}, or \cite[Lemma 17]{kueng_low_2015}).

For a corresponding upper bound on $\mathbb{E} \left[ S^4 \right]$, the 4-design property of the POVM is of crucial importance. Without requiring further assumptions, Eq.~\eqref{eq:4design} assures
\begin{align}
 \mathbb{E} \left[ S_X^4 \right] 
=& \mathrm{tr} \left( \frac{1}{N} \sum_{k=1}^N \left( |x_k \rangle \! \langle x_k | \right)^{\otimes 4} X^{\otimes 4} \right) \nonumber \\
=& \binom{d+3}{4}^{-1} \mathrm{tr} \left( P_{\mathrm{Sym}^4} X^{\otimes 4} \right) \nonumber \\
\leq& \frac{10.1\left( \mathrm{tr} \left( X^2 \right) + \mathrm{tr}(X)^2\right)^2}{d(d+1)(d+2)(d+3)} \nonumber\\
= & \frac{10.1 d(d+1) }{(d+2)(d+3)}\mathbb{E} \left[ S^2 \right]^2, \label{eq:4moment_bound}
\end{align}
where the inequality follows from \autoref{lem:4mIneq} in the appendix. Here we content ourselves to state that standard techniques such as \cite[Lemma 17]{kueng_low_2015} allow for evaluating 
$\mathrm{tr} \left( P_{\mathrm{Sym}^4} X^{\otimes 4} \right)$ explicitly without requiring $X$ to have vanishing trace. Similar techniques were also employed in \cite{lancien_distinguishing_2013}.
Earlier approaches, such as Refs.~\cite{ambainis_quantum_2007,matthews_distinguishability_2009}, made the assumption $\mathrm{tr}(X)=0$ to considerably simplify the evaluation of  $\mathbb{E} \left[ S_X^4 \right]$.
Inserting these bounds into \eqref{eq:4design_aux1} reveals
\begin{align*}
& \| \mathcal{M}_{\textrm{4d}}(X) \|_{\ell_1}
\geq d \sqrt{ \frac{ \mathbb{E} \left[ S^2 \right]^3}{\mathbb{E} \left[ S^4 \right]}}\\
=&\sqrt{\frac{(d+2)(d+3)}{(d+1)^2}
		\frac{(\|X\|_2^2+\tr(X)^2)^3}{24\|X\|_2^2 \tr(P_{\mathrm{Sym}^4} X^{\otimes 4}) }}\|X\|_2\nonumber\\
\geq&\sqrt{\frac{(\|X\|_2^2+\tr(X)^2)^3}{24\|X\|_2^2 \tr(P_{\mathrm{Sym}^4} X^{\otimes 4}) }}\|X\|_2\nonumber\\		
\geq&  \frac{1}{\sqrt{9.673}} \| X \|_2
> \frac{0.32}{\sqrt{ \mathrm{rank}(X)}} \| X \|_1,
\end{align*}
where the third inequality follows from \autoref{lem:4mIneq2} in the appendix. 
Since the choice of $X \in \mathcal{H}_d$ is arbitrary, \eqref{eq:4design_bound} in  \autoref{thm:4design} readily follows.

To derive the tighter bound valid for rank-two matrices note that
\begin{align*}
\| \mathcal{M}_{\textrm{4d}}(X) \|_{\ell_1}
&\geq\sqrt{\frac{(\|X\|_2^2+\tr(X)^2)^3}{24\|X\|_1^2 \tr(P_{\mathrm{Sym}^4} X^{\otimes 4}) }}\|X\|_1\nonumber\\		
&\geq  \frac{1}{\sqrt{12.12}} \| X \|_1
> \frac{0.4 \| X \|_1}{\sqrt{\mathrm{rank}(X)}},
\end{align*} 
where the second inequality follows from \autoref{lem:4mIneq3} in the appendix.
If $X$ is both rank two and traceless, \autoref{lem:4mIneq3} actually implies
\begin{align}
\frac{24\| X \|_1^2\tr \left( P_{\mathrm{Sym}^4} X^{\otimes 4} \right)}{ [\| X \|_2^2 + \tr (X)^2]^3}=12.
\end{align}
This allows for further improving the constant $\frac{1}{\sqrt{12.12}}$ to $\frac{1}{2\sqrt{3}}$.

\subsection{A bound for Clifford POVMs}

Now let us move on to prove \autoref{thm:main}---a similar statement for Clifford POVMs. 
Fix $d=2^n$ and let $\mathcal{M}_{C,z} = \left\{ \frac{d}{N} |x_k \rangle \! \langle x_k |\right\}_{k=1}^N$ be a Clifford orbit POVM with fiducial $|z\rangle \in \mathbb{C}^d$.
We fix $X \in \mathcal{H}_d$ and define the random variable $S_X$ in analogy to \eqref{eq:S}. Similar to before, doing so assures
\begin{equation*}
\| \mathcal{M}_{C,z}(X) \|_{\ell_1} = d \mathbb{E} \left[ |S_X| \right] \geq d \sqrt{\frac{ \mathbb{E} \left[ S^2_X \right]^3}{\mathbb{E} \left[ S^4_X \right]}}
\end{equation*}
via Berger's inequality.
As already pointed out in Sec.~\ref{sec:results}, any Clifford orbit does constitute a complex projective 3-design. This in turn implies that \eqref{eq:2moment_bound} remains valid, because its derivation just requires a 2-design structure:
\begin{equation}
\mathbb{E} \left[ S^2_X \right] = \frac{ \| X \|_2^2 + \tr (X)^2}{(d+1)d}. \label{eq:clifford_2moment}
\end{equation}
However, deriving a corresponding bound for $\mathbb{E} \left[ S_X^4 \right]$ is considerably more challenging. This is because Clifford orbits in general fall short of being complex projective 4-designs.
Instead, we resort to Eq.~\eqref{eq:main_technical} which implies
\begin{align}
&\mathbb{E} \left[ S_X^4 \right]
= \tr \left( \frac{1}{N} \left( |x_k \rangle \! \langle x_k | \right)^{\otimes 4} X^{\otimes 4} \right) \nonumber  \\
=& d \binom{d+2}{3}^{-1}
\left( \left( \alpha (z) -\beta (z) \right) \tr \left( P_1 X^{\otimes 4} \right) \right. \nonumber \\
+& \left. \beta (z) \tr \left( P_{\mathrm{Sym}^4} X^{\otimes 4} \right) \right), \label{eq:main_aux1}
\end{align}
where $P_1 \in \mathcal{H}_d^{\otimes 4}$ and $\beta (z)$ were introduced in \autoref{thm:main_technical}.
A bound on $\mathrm{tr} \left( P_{\mathrm{Sym}^4} X^{\otimes 4} \right)$ was already obtained in the previous subsection, see \eqref{eq:4moment_bound}.
For the remaining term, we obtain
\begin{align*}
\tr (P_1 X^{\otimes 4})
=& \tr \left( P_1 \left| X^{\otimes 4} \right| \right)
= \tr \left( P_{\mathrm{Sym}^4} Q |X|^{\otimes 4} \right) \\
\leq & \tr \left( Q |X|^{\otimes 4} \right)
= \frac{1}{d^2} \sum_{k=1}^{d^2} \tr \left( W_k^{\otimes 4} |X|^{\otimes 4} \right) \\
=& \frac{1}{d^2} \sum_{k=1}^{d^2} \tr \left( W_k |X| \right)^4
\end{align*}
by invoking some standard trace inequalities.
Hoelder's inequality together with the fact that the characteristic function \eqref{eq:characteristic_function} is proportional to an isometry ($\| \Xi (X) \|_{\ell_2}^2 = d \| X \|_2$) allows us to simplify further:
\begin{align}
\tr (P_1 X^{\otimes 4})
\leq &  \frac{1}{d^2} \sum_{k=1}^{d^2} \tr \left( W_k |X| \right)^4 \nonumber \\
\leq&  \frac{1}{d^2} \sum_{k=1}^{d^2} \| X \|_1^2 \| W_k \|_\infty^2 \tr \left( W_k |X| \right)^2 \nonumber \\
=& \frac{ \| X \|_1^2}{d^2} \| \Xi (|X|) \|_{\ell_2}^2
= \frac{ \| X \|_1^2 \| X \|_2^2}{d}. 
\end{align}
The last equation is due to the fact that the Schatten-$p$ norms of $X$ and that of $|X|$ coincide by definition. Together with \eqref{eq:main_aux1}, this relation implies  
\begin{align}
\mathbb{E} \left[ S_X^4 \right]
\leq &  \binom{d+2}{3}^{-1} |\alpha (z)-\beta (z)| \| X \|_1^2 \| X \|_2^2 \nonumber \\
+& \frac{24}{(d+4)(d+1)^2 d} \tr \left( P_{\mathrm{Sym}^4} X^{\otimes 4} \right), \label{eq:clifford_4moment}
\end{align}
where we have used
\begin{align*}
d \binom{d+2}{3}^{-1} \beta (z) =& \frac{24(1- \alpha (z))}{(d+4)(d+2)(d+1)(d-1)} \\
\leq & \frac{24}{(d+4)(d+1)^2 d},
\end{align*}
which is due to \eqref{eq:alpha1_bound}.
Combining this fourth moment bound \eqref{eq:clifford_4moment} with the second moment bound from \eqref{eq:clifford_2moment} implies
\begin{align}
\| \mathcal{M}_{C,z}(X) \|_{\ell_1}
\geq &  d \sqrt{ \frac{ \mathbb{E} \left[ S^2_X \right]^3}{\mathbb{E} \left[ S^4_X \right]}}  
\geq  \frac{\| X \|_2}{\sqrt{\kappa(X,z)}}, \label{eq:POVMnormConstants}
\end{align}
where
\begin{widetext}
\begin{align}
\kappa(X,z)=
&\frac{\frac{6(d+1)^2}{d+2}|\alpha (z)-\beta (z)| \| X \|_1^2 \| X \|_2^4+\frac{24(d+1)}{d+4 }\| X \|_2^2\tr \left( P_{\mathrm{Sym}^4} X^{\otimes 4} \right)}{ (\| X \|_2^2 + \tr (X)^2)^3} \nonumber\\
\leq& \frac{\frac{6d(d+1)^2}{(d+2)(d+4)}\alpha (z) \| X \|_1^2 \| X \|_2^4+\frac{24(d+1)}{d+4 }\| X \|_2^2\tr \left( P_{\mathrm{Sym}^4} X^{\otimes 4} \right)}{ (\| X \|_2^2 + \tr (X)^2)^3} \nonumber \\
\leq & \frac{d+1}{d+4}\left(6d\alpha (z) \frac{\| X \|_1^2 }{\| X \|_2^2}+9.673\right)
\leq 6d \alpha (z) r_{\mathrm{eff}}(X)+10. \label{eq:kappa}
\end{align}
\end{widetext}
Here, the second inequality follows from  \autoref{lem:4mIneq2} in the appendix and the last one exploits the definition of effective rank: $r_{\mathrm{eff}}(X) = \frac{ \| X \|_1^2}{\| X \|_2^2}$.
Inserting this bound into \eqref{eq:POVMnormConstants} yields \autoref{thm:main}.

As already pointed out in Sec.~\ref{sec:results}, typical random fiducials $|z\rangle \in \mathbb{C}^d$ obey 
\begin{equation}
\alpha (z) \leq \frac{6}{(d+3)d}, \label{eq:generic_characteristic}
\end{equation}
see also \cite{other_paper}. Such a constraint allows for a considerable improvement:

\begin{proposition} \label{thm:main_generic}
	Fix $d=2^n$ and let $\mathcal{M}_{C,z}$ be a Clifford POVM whose fiducial $|z \rangle \in \mathbb{C}^d$ obeys \eqref{eq:generic_characteristic}. Then
	\begin{align*}
\| \mathcal{M}_{C,z}(X) \|_{\ell_1}
	\ge \frac{\| X \|_1}{\sqrt{22r_{\mathrm{eff}}(X)}}
	\end{align*}
	for any $X \in \mathcal{H}_d$. 
\end{proposition}

\begin{proof}
Inequality~\eqref{eq:generic_characteristic} assures 
\begin{equation*}
-\frac{2}{d(d+1)}\leq \alpha (z)-\beta(z)\leq \frac{2}{(d-1)(d+4)}.
\end{equation*}
In turn, $\kappa (X,z)$ featuring in \eqref{eq:POVMnormConstants} may be bounded by
\begin{align*}
\kappa(X,z) 
\leq &  \frac{ \frac{12\| X \|_1^2}{d\| X \|_2^2}\| X \|_2^6+24\| X \|_2^2\tr \left( P_{\mathrm{Sym}^4} X^{\otimes 4} \right)}{ [\| X \|_2^2 + \tr (X)^2]^3} \\
\leq & \frac{12 r_{\mathrm{eff}} (X)}{d} + 10 \leq 22,
\end{align*}
and the statement readily follows.
\end{proof}

\subsection{Proof of the converse bound: \autoref{thm:converse}}

\autoref{thm:converse} provides a converse bound to \autoref{thm:main}.
At the heart of its proof is the fact that by definition the multi-qubit Clifford group is the normalizer of the Pauli group $P(d)=\left\{ \pm W_k, \pm i W_k \right\}_{k=1}^{d^2}$ and it acts transitively on Pauli operators up to overall phase factors. 
This fact in particular implies that 
\begin{align*}
\| \mathcal{M}_{C,z}(W) \|_{\ell_1}
=& \frac{d}{|\mathrm{C}_n|} \sum_{j=1}^{|\mathrm{C}_n|} \left| \langle C_j z | W| C_j z\rangle \right| \\
=& \frac{d}{|\mathrm{C}_n|} \sum_{j=1}^{|\mathrm{C}_n|} \left| \langle z| C_j^\dagger W C_j |z \rangle \right| \\
=& \frac{d}{d^2-1} \sum_{k=2}^{d^2} \left| \langle z |W_k | z \rangle \right|.
\end{align*}
Using $\langle z | W_1 | z \rangle = \langle z | z \rangle =1$ and the definition  of the characteristic function in \eqref{eq:characteristic_function}, this expression amounts to
\begin{align*}
\| \mathcal{M}_{C,z}(W) \|_{\ell_1}
=& \frac{ d(\sum_{k=1}^{d^2} \left| \mathrm{tr} \left( W_k |z \rangle \! \langle z | \right) \right| -1) }{d^2-1} \\
=& \frac{ d \left( \| \Xi (|z \rangle \! \langle z|) \|_{\ell_1} -1 \right)}{d^2-1} \\
=& \frac{ \|\Xi (|z \rangle \! \langle z| ) \|_{\ell_1}-1}{(d+1)(d-1)} \| W \|_1,
\end{align*}
because $\| W \|_1 = d$ for any Pauli matrix.

This pre-factor can be related to $\alpha (z)$ which is the main figure of merit in \autoref{thm:main}. Indeed,
\begin{equation*}
 \|\Xi (|z \rangle \! \langle z| ) \|_{\ell_1} \geq 
 \sqrt{\frac{ \|\Xi (|z \rangle \! \langle z| ) \|_{\ell_2}^6}{ \|\Xi (|z \rangle \! \langle z| ) \|_{\ell_4}^4}}=\frac{d^{\frac{3}{2}}}{ \|\Xi (|z \rangle \! \langle z| ) \|_{\ell_4}^2   },
\end{equation*}
because the characteristic function is proportional to an isometry. This in turn implies
\begin{equation}\label{eq:converse3}
 \frac{\|\Xi (|z \rangle \! \langle z| ) \|_{\ell_1}-1}{(d+1)(d-1)}\geq  \frac{ \frac{d^{\frac{3}{2}}}{\alpha (z)}-1}{(d+1)(d-1)}.
 \end{equation}

\subsection{Characteristic function of different fiducials and their implications} \label{sec:fiducials}

The characteristic functions of stabilizer states are well-known \cite{Gros06}. Nonetheless, we shall  derive it here for the sake completeness.
In dimension $d=2^n$, every stabilizer state $|z \rangle \in \mathbb{C}^d$ is a common eigenvector of an order-$d$ Abelian subgroup of the Pauli group $P(d) = \left\{ \pm W_k, \pm i W_k \right\}_{k=1}^d$ that does not contain $-\mathbb{I}$. This in turn implies that (see e.g.\ \cite[Exercise 10.34]{nielsen_quantum_2010})
\begin{equation*}
|z \rangle \! \langle z| = \frac{1}{d} \sum_{k \in S} \phi_k W_k \quad \phi_k \in \left\{ \pm 1 \right\}.
\end{equation*} 
Here $S \subset \left\{1,\ldots,d^2 \right\}$ is a subset of cardinality $|S|=d$. Mutual orthogonality of the Pauli matrices with respect to the Hilbert-Schmidt inner product then implies
\begin{align*}
\Xi (|z \rangle \! \langle z| )
=& \sum_{j=1}^{d^2} \tr \left( W_j \frac{1}{d} \sum_{k \in S} \phi_k W_k \right) |e_j \rangle \\
=& \frac{1}{d} \sum_{j=1}^{d^2} \sum_{k \in S} \phi_k \mathrm{tr} \left( W_k W_j \right) |e_j \rangle \\
=& \sum_{k \in S} \phi_k |e_k \rangle.
\end{align*}
Accordingly,
\begin{equation}
\| \Xi (|z \rangle \! \langle z|) \|_{\ell_p}^p =  \sum_{k \in S} \left| \phi_k \right|^p = d \label{eq:stab_characteristic_function}
\end{equation}
for any $1 \leq p < \infty$. 

Since $\alpha (z) = \frac{1}{d^2} \| \Xi (|z \rangle \! \langle z|) \|_{\ell_4}^4$ this in particular implies $d \alpha_1 (z) = 1$ and consequently
\begin{equation*}
\kappa (X,z ) \leq \frac{d+1}{d+4} \left( 6\frac{ \| X \|_1^2}{\| X \|_2^2} + 10 \right) \leq 6d, \quad \forall X \in \mathcal{H}_d
\end{equation*}
where $\kappa (X,z)$ was defined in \eqref{eq:kappa}. This in turn implies
\begin{equation*}
\|\mathcal{M}_{\mathrm{stab}}( X) \|_{\ell_1} \geq \frac{ \| X \|_1}{\sqrt{ \kappa (X,z) \mathrm{rank}(X)}} \geq \frac{ \| X \|_1}{\sqrt{6} d} \quad \forall X \in \mathcal{H}_d
\end{equation*}
which confirms the lower bound in \eqref{eq:stab_bound}.

For rank-two matrices $X \in \mathcal{H}_d$ and stabilizer state fiducials ($| \alpha (z) - \beta (z) |= \frac{1}{d+4}$)  the bound on $\kappa (X,z)$ in \eqref{eq:POVMnormConstants} may be further simplified to
\begin{align*}
\kappa (X,z) \leq & \frac{6 \| X \|_1^2 \| X \|_2^4 + 24 \| X \|_1^2 \mathrm{tr} \left( P_{\mathrm{Sym}^4} X^{\otimes 4} \right)}{\left( \| X \|_2^2 + \mathrm{tr}(X)^2 \right)^3}
\leq 36,
\end{align*}
where the last inequality is due to \autoref{lem:4mIneq4} in the appendix. This in turn implies
\begin{equation*}
\|\mathcal{M}_{\mathrm{stab}}( X) \|_{\ell_1} \geq \frac{1}{6} \| X \|_1 \quad \forall X \in \mathcal{H}_d:\; \mathrm{rank}(X) = 2.
\end{equation*}
\autoref{prop:pure_optimal} is an immediate consequence from this. This statement is in fact valid for arbitrary Clifford orbits, because stabilizer states lead to a worst case behavior of $\kappa (X,z)$.

Finally, Eq.~\eqref{eq:stab_characteristic_function} also implies that the constant in \autoref{thm:converse} amounts to
\begin{equation*}
\frac{ \| \Xi (|z \rangle \! \langle z|) \|_{\ell_1}-1}{(d+1)(d-1)} = \frac{1}{d+1}
\end{equation*}
which confirms that the lower bound presented in \eqref{eq:converse2} is indeed saturated for stabilizer states.

Let us now turn our attention to the characteristic function of the ``magic product state'' $|z \rangle \! \langle z|=\rho^{\otimes n} \in\mathcal{H}_{2^n}$ with $\rho = \frac{1}{2} \left(\sigma_0 + \frac{1}{\sqrt{3}} \left( \sigma_1+\sigma_2+\sigma_3 \right) \right) \in \mathcal{H}_2$. Here $\sigma_0,\ldots,\sigma_3 \in \mathcal{H}_2$ denote the single qubit Pauli matrices with the convention $\sigma_0 = \mathbb{I}$. 
We will content ourselves with directly computing $\ell_p$ norms of the characteristic function. 
To this end, we use the fact that every $d=2^n$-dimensional Pauli matrix admits a tensor product decomposition
\begin{equation*}
W_k = \sigma_{k_1} \otimes \cdots \otimes \sigma_{k_n} \quad k_j \in \left\{0,1,2,3 \right\}
\end{equation*}
into single qubit Pauli's. Doing so implies
\begin{align*}
\| \Xi ( \rho^{\otimes n} ) \|_{\ell_p}^p
=& \sum_{k_1,\ldots,k_n=0}^3 \left| \tr \left( W_{k_1} \otimes \cdots \otimes W_{k_n} \rho^{\otimes n}  \right) \right|^p \\
=& \sum_{k_1,\ldots,k_n=0}^3 \left| \tr \left( W_{k_1} \rho \right) \cdots \tr \left(W_{k_n} \rho \right) \right|^p \\
=& \prod_{j=1}^n \sum_{k_j=0}^3 \left|\tr \left( W_{k_j} \rho \right)\right|^p \\
=& \prod_{j=1}^n \left( 1 + 3 \left( \frac{1}{\sqrt{3}} \right)^p  \right) \\
=& \left( 1 + 3 \left( \frac{1}{\sqrt{3}} \right)^p  \right)^n.
\end{align*}
For $\alpha (z)$ defined in \eqref{eq:alpha}, we thus obtain
\begin{align}
\alpha (z) =& \frac{1}{d^2}\| \Xi (|z \rangle\! \langle z| ) \|_{\ell_4}^4
= \frac{1}{d^2}\left( 1 + \frac{3}{9} \right)^n = \frac{4^n}{2^{2n} 3^n} \nonumber \\
=& \frac{1}{3^n} = \left( \frac{1}{9}\right)^{\frac{n}{2}}
< \left(\frac{1}{8} \right)^{\frac{n}{2}} = d^{-\frac{3}{2}}. \label{eq:magic_4norm}
\end{align}
Inserting this into \autoref{thm:main} leads to relations~\eqref{eq:magic_good} and \eqref{eq:magic}. 
Similarly:
\begin{equation}
\| \Xi (|z \rangle \! \langle z| ) \|_{\ell_1}
= \left( 1 + \sqrt{3} \right)^n < d^{1.45}, \label{eq:magic_1norm}
\end{equation}
and inserting this into \autoref{thm:converse} implies the converse bound \eqref{eq:magic_converse} for Clifford orbits with a magic state fiducial.

\section{Entropic uncertainty and certainty relations} \label{sec:entropic_proofs}

Let us start this section with re-capitulating a proof of the strong average uncertainty relations for both mutually unbiased bases \eqref{eq:mub_uncertainty} and stabilizer bases \eqref{eq:stab_uncertainty}.
As pointed out in \cite{wehner_entropic_2010}, both statements follow from the fact that complete sets of mutually unbiased bases and stabilizer bases, respectively, form complex projective 2-designs in prime power dimensions.

It is instructive to repeat their argument.
First note that every quantum state $\rho \in \mathcal{S}_d$ amounts to a convex combination of pure states. Since entropy is concave, Jensen's inequality allows for restricting our attention to pure states $\phi = |\phi \rangle \! \langle \phi | \in \mathcal{S}_d$.

Next, we point out that the Shannon entropy is a special case of a more general family of entropy functions: \emph{R\'enyi entropies}. Let $p \in \mathbb{R}^N$ be a discrete probability distribution represented by a probability vector. Then for every $\alpha \geq 0$ 
the R\'enyi entropy of this distribution is defined as
\begin{equation}
H_\alpha \left( p \right) = \frac{1}{1-\alpha} \log_2 \left( \left\| p \right\|_{\ell_\alpha}^\alpha \right). \label{eq:renyi}
\end{equation}
These R\'enyi entropies are monotonically decreasing in $\alpha$, i.e. $H_\alpha (p ) \geq H_\beta (p)$ for any $p$ provided that $\alpha \leq \beta$.
The Shannon entropy \eqref{eq:shannon} arises from taking the limit $\alpha \to 1$ in \eqref{eq:renyi}. 
The R\'enyi entropy of order $\alpha =2$---also known as \emph{collision entropy}---provides a lower bound on the Shannon entropy: $H (p) \geq H_2 (p)$.

Now, let $\mathcal{B}_1,\ldots,\mathcal{B}_{M}$ denote a family of orthonormal bases whose union forms a complex projective 2-design of cardinality $N=dM$: 
\begin{equation*}
\mathcal{M}_{\mathrm{2D}} = \left\{ |b_1^{(k)}\rangle,\ldots,|b_d^{(k)} \rangle \right\}_{k=1}^M = \left\{ x_k \right\}_{k=1}^N.
\end{equation*}
A complete set of mutually unbiased bases, as well as stabilizer states, form particular instances of such configurations. 
Then, lower bounding the Shannon entropy by the collision entropy and exploiting concavity of the logarithm results in
\begin{align}
\frac{1}{M} \sum_{k=1}^M H \left( \mathcal{B}_k | \phi \right)
\geq & \frac{1}{M} \sum_{k=1}^M H_2 \left( \mathcal{B}_k (\phi) \right) \nonumber \\
=& \frac{1}{M} \sum_{k=1}^M - \log_2 \left( \sum_{j=1}^d  \langle b_j^{(k)} | \phi | b_j^{(k)} \rangle^2 \right) \nonumber \\
\geq & - \log_2 \left( \frac{d}{N} \sum_{k=1}^N \langle x_k | \phi |x_k \rangle^2 \right) \nonumber \\
=& - \log_2 \left( d \mathbb{E} \left[ S_\phi^2 \right] \right). \label{eq:uncertainty_aux1}
\end{align}
Here, we have rewritten $\frac{1}{N}\sum_{k=1}^N \langle x_k | \phi | x_k \rangle^2$ as the second moment of the random variable $S_\phi$ defined in \eqref{eq:S}.
Since, $\mathcal{M}_{\mathrm{2D}}$ forms a complex projective 2-design, formula~\eqref{eq:2moment_bound} is valid and implies
\begin{equation*}
\mathbb{E} \left[ S_\phi^2 \right] = \frac{ \mathrm{tr}(\phi)^2 + \mathrm{tr}(\phi^2)}{(d+1)d} = \frac{2}{(d+1)d},
\end{equation*}
because $\phi \in \mathcal{S}_d$ is pure. Inserting this into \eqref{eq:uncertainty_aux1} allows us to conclude
\begin{equation*}
\frac{1}{M} \sum_{k=1}^M H \left( \mathcal{B}_k | \phi \right)
\geq - \log_2 \left( \frac{2}{d+1} \right) = \log (d+1)-1,
\end{equation*}
as claimed.

Such a proof strategy in principle also allows for taking into account design properties of higher order. Indeed, suppose that the union of $\mathcal{B}_1,\ldots,\mathcal{B}_M$ forms a complex projective $t$-design with $t \geq 3$.
Then lower bounding the Shannon entropy with the R\'enyi entropy of order $t$ instead of the collision entropy results in
\begin{equation*}
\frac{1}{M} \sum_{k=1}^M H \left( \mathcal{B}_k | \phi \right) \geq \frac{1}{1-t} \log_2 \left( \frac{t! d!}{(t+d-1)!} \right).
\end{equation*}
As pointed out in \cite{wehner_entropic_2010}, this bound becomes weaker as $t$ increases. 

On first sight, this prevents us from exploiting the additional information about third and fourth moments of stabilizer states obtained in this work.
For stabilizer states, we do know the first four moments of the random variable $S_\phi$ exactly, because $\phi$ is pure. 
In order to exploit this additional information, we formalize a linear programming approach which is inspired by \cite{serafini_canonical_2007}.
The key idea is to replace the discrete random variable $S_\phi \in [0,1]$ by its density function
\begin{equation*}
\mu_{\mathrm{stab},\phi} (x) = \frac{1}{N} \sum_{k=1}^N \delta \left( x - \langle x_k | \phi |x_k \rangle \right)
\end{equation*}
on the unit interval $[0,1]$. Doing so allows us to formulate arbitrary moments of order $\alpha \geq 0$ as integrals
\begin{equation*}
\mathbb{E} \left[ S_\phi^\alpha \right] = \int_0^1 \mu_{\mathrm{stab},\phi}(x) x^\alpha \mathrm{d}x 
\end{equation*}
which are linear in $\mu_{\mathrm{stab},\phi}$. The fact that stabilizer states form a complex projective 3-design completely specify the moments of $\mu_{\mathrm{stab},\phi}$ for $\alpha=1,2,3$. 
Moreover, \autoref{thm:main_technical} puts an upper bound on the fourth moment which we derive in the appendix.

For any $1 < \alpha <2$ we may thus obtain an upper bound on $\mathbb{E} \left[ S_\phi^\alpha \right]$ by solving the following maximization over probability densities $\mu: [0,1] \to [0,1]$:
\begin{align}
\underset{ \mu:[0,1]\to \mathbb{R}}{\textrm{maximize}} & \quad \int_0^1 \mu (x) x^\alpha \mathrm{d} x \label{eq:measure_maximization}\\
\textrm{subject to} & \quad \int_0^1 \mu(x) x \mathrm{d} x = \frac{1}{d}, \nonumber \\
& \quad \int_0^1 \mu(x) x^2 \mathrm{d} x = \frac{2}{(d+1)d}, \nonumber \\
& \quad \int_0^1 \mu(x) x^3 \mathrm{d} x = \frac{6}{(d+2)(d+1)d}, \nonumber \\
&\quad \int_0^2 \mu (x) x^4 \mathrm{d} x \leq \frac{30}{(d+4)(d+2)(d+1)d} \nonumber \\
&\quad  \int_0^1 \mu (x) \mathrm{d}x = 1,\quad \mu (x) \geq 0 \quad \forall x \in [0,1]. \nonumber
\end{align}
The first four constraints demand that $\mu$ reproduces the moments of $\mu_{\mathrm{stab},\phi}$ for $\alpha=1,2,3,4$. On the contrary to the previous approach, these constraints do take into account additional information about the higher moments of stabilizer states.
The final constraints simply enforce $\mu$ to be a valid probability density on the unit interval.

In order to render this optimization computationally tractable, we coarse-grain the unit interval $[0,1]$ to a ``grid'' of $D$ elements: $\vec{g}= \left( 0, \frac{1}{D-1},\frac{2}{D-1},\ldots,\frac{D-2}{D-1},1 \right)^T \in \mathbb{R}^D$.
We also introduce ``powers'' of this grid vector in order to represent the constraints in \eqref{eq:measure_maximization}:
\begin{equation*}
\vec{g}_\beta := \left( 0, \left( \frac{1}{D-1} \right)^\beta, \left( \frac{2}{D-1} \right)^\beta, \ldots, \left( \frac{D-2}{D-1} \right)^\beta, 1 \right).
\end{equation*}
with $\beta = \alpha,2,3,4$ and $\vec{g}_1 = \vec{g}$.
Likewise, we represent each $\mu$ by a $D$-dimensional probability vector $\vec{\mu} \in \mathbb{R}^D$ on this grid.
Such a discretization approximates \eqref{eq:measure_maximization} by
\begin{align}
\lambda_\alpha := \underset{\vec{\mu} \in \mathbb{R}^D}{\textrm{maximize}}
& \quad \langle \vec{\mu}, \vec{g}_\alpha \rangle \label{eq:linear_program} \\
\textrm{subject to} & \quad \langle \vec{\mu}, \vec{g} \rangle = \frac{1}{d}, \nonumber \\
& \quad \langle \vec{\mu}, \vec{g}_2 \rangle = \frac{2}{(d+1)d}, \nonumber \\
&\quad \langle \vec{\mu}, \vec{g}_3 \rangle = \frac{6}{(d+2)(d+1)d}, \nonumber \\
& \quad \langle \vec{\mu}, \vec{g}_4 \rangle \leq \frac{30}{(d+4)(d+2)(d+1)d}, \nonumber \\
& \quad \langle \vec{\mu}, \vec{1} \rangle = 1, \quad \vec{\mu} \geq \vec{0} \nonumber
\end{align}
which is a linear program. Here, $\vec{1} =(1,\ldots,1)^T \in \mathbb{R}^D$ denotes the ``all-ones'' vector and $\vec{\mu} \geq \vec{0}$ indicates non-negativity of $\vec{\mu}$ in the sense that all its vector components are non-negative.
In particular, we can conclude
\begin{equation*}
\mathbb{E} \left[ S_\phi^\alpha \right] \leq \lambda_\alpha \quad \forall \phi =| \phi \rangle \! \langle \phi| \in \mathcal{S}_d.
\end{equation*}

In order to obtain an entropic uncertainty relation, we fix $\alpha=1+\epsilon$ close to one (e.g. $\epsilon=0.1$), solve the linear program \eqref{eq:linear_program} for this value of $\alpha$ and then use monotonicity of R\'enyi entropies, as well as concavity of the logarithm, in a fashion similar to before:
\begin{align*}
\frac{1}{M} \sum_{k=1}^M H \left( \mathcal{B}_k | \phi \right)
\geq & \frac{1}{M} \sum_{k=1}^M H_{\alpha} \left( \mathcal{B}_k | \phi \right)  \\
\geq & \frac{1}{1-\alpha} \log_2 \left( \frac{d}{N} \sum_{k=1}^N \langle x_k | \phi |x_k \rangle^\alpha \right) \\
=& \frac{1}{1-\alpha} \log_2 \left( d \mathbb{E} \left[ S_\phi^\alpha \right] \right) \\
\geq & \frac{1}{1-\alpha} \log_2 \left( d \lambda_\alpha \right).
\end{align*}

Analogous approaches work for 2-,3- and 4-designs, provided that one adjusts the constraints in the linear program \eqref{eq:linear_program} appropriately. 
\autoref{fig:uncertainty} compares the different results obtained in such a way graphically over a wide range of dimensions $d=2^n$, where our results about stabilizer states apply.

For the present paper, we do content ourselves with these numerically obtained stronger uncertainty relations for stabilizer states. However, linear programs are a very versatile theoretical tool and it is highly plausible that a more detailed analysis will allow for supporting our numerical findings with analytical proofs. We leave this to future work. 

Let us now turn our attention to converse bounds, so-called \emph{certainty relations}. We have presented two such statements in Sec.~\ref{sub:entropic_main}. Both results closely resemble certainty relations in \cite[Section 5]{matthews_distinguishability_2009} that were formalized for complex projective 2- and 4-designs, respectively.

For the first result, we consider an isotropic ensemble $\sum_{x} p_x \phi_x = \frac{1}{d} \mathbb{I} \in \mathcal{S}_d$ of pure states and a Clifford POVM measurement $\mathcal{M}_{C,z}: \mathcal{H}_d \to \mathbb{R}^N$.  
Then, the Shannon mutual information between $X$, the preparation variable, and the measurement outcome of $\mathcal{M}_{C,z}$ may be bounded by
\begin{align*}
I  \left( X: \mathcal{M}_{C,z} \right)
=& \sum_x p_x D \left( \mathcal{M}_{C,z} (\phi_x) \| \mathcal{M}_{C,z} \left( \frac{1}{d} \mathbb{I} \right) \right) \\
\geq & \sum_x \frac{p_x}{2 \log (2)} \left\| \mathcal{M}_{C,z} (\phi_x) - \mathcal{M}_{C,z} \left( \frac{1}{d} \mathbb{I} \right) \right\|_{\ell_1}^2 \\
=& \sum_x \frac{p_x}{2 \log (2)} \left\| \mathcal{M}_{C,z} \left( \phi_x - \frac{1}{d} \mathbb{I} \right)\right\|_{\ell_1}.
\end{align*}
This bound follows from Pinsker's inequality $D(\rho \| \sigma) \geq \frac{1}{2 \log (2)} \left\| \rho - \sigma \right\|_1^2$. 
Now, we can use the fact that $X:= \phi_x - \frac{1}{d} \mathbb{I}$ has ``effective rank'' four. This is a consequence of the following lemma.

\begin{lemma} \label{lem:effective_rank}
Let $\rho \in \mathcal{S}_d$ be quantum state with $\mathrm{rank}(\rho)=r$. Then the ``effective rank'' of $X = \rho - \frac{1}{d} \mathbb{I}$ amounts to
\begin{align*}
r_{\mathrm{eff}}(X) :=& \frac{ \|X\|_1^2}{\|X\|_2^2} \leq \frac{4 \mathrm{rank}(\rho)(d-\mathrm{rank}(\rho))}{d} \\
\leq& 4 \min \left\{ r, d- r \right\}.
\end{align*}
The first bound is saturated by quantum states $\rho$ that are maximally mixed on an $r$-dimensional subspace, while the second bound is saturated, if $\rho$ is pure.
\end{lemma}

We provide a proof of this statement in the appendix.
As explained in Sec.~\ref{sub:main_distinguishability}, \autoref{thm:main} remains valid if we replace the actual rank of a matrix by its ``effective rank''. 
Among all possible Clifford orbit POVMs, \autoref{thm:main} is weakest for stabilizer states. The ``effective rank'' reformulation of the corresponding bound \eqref{eq:main_stab} reads
$\| \mathcal{M}_{\mathrm{stab}}(X) \|_{\ell_1}\geq \frac{\| X \|_2^2}{4 \|X \|_1^2} \| X \|_1$ and we obtain
\begin{align*}
I  \left( X: \mathcal{M}_{C,z} \right)
\geq & \sum_x \frac{p_x}{2 \log (2)} \left\| \mathcal{M}_{C,z} \left( \phi_x - \frac{1}{d} \mathbb{I}\right) \right\|_{\ell_1}^2\\
\geq & \sum_x \frac{p_x}{512 \log (2)} \left\| \phi_x - \frac{1}{d} \mathbb{I} \right\|_1^2 \\
=& \frac{1}{128 \log (2)} \left( \frac{d-1}{d} \right)^2,
\end{align*}
because $\left\| \phi_x - \frac{1}{d} \mathbb{I} \right\|_1 = 2\frac{d-1}{d}$ and $\sum_x p_x = 1$. 
This bound holds for arbitrary Clifford POVMs including stabilizer states. 
However, the constant $\frac{1}{128}$ may be further improved for particular Clifford orbits, such as 4-designs. 

In order to derive the second certainty relation \eqref{eq:certainty2}, we once more follow a similar calculation presented in \cite{matthews_distinguishability_2009}. 
Let $\mathcal{M}_{C,z}: \mathcal{H}_d \to \mathbb{R}^N$ denote a Clifford POVM measurement that maps states to probability vectors.
Then the definition of the relative entropy together with Pinsker's inequality imply for any pure state $\phi \in \mathcal{S}_d$
\begin{align*}
\log (N) - S \left( \mathcal{M}_{C,z} | \phi \right) 
=& D \left( \mathcal{M}(\phi) \| \mathcal{M}_{C,z} \left( \frac{1}{d} \mathbb{I} \right) \right) \\
\geq & \frac{1}{2 \log (2)} \left\| \mathcal{M}_{C,z} \left( \phi - \frac{1}{d} \mathbb{I} \right) \right\|_{\ell_1}^2 \\
\geq & \frac{1}{512 \log (2)} \left\| \phi - \frac{1}{d} \mathbb{I} \right\|_1^2 \\
=& \frac{1}{128 \log (2)} \left( \frac{d-1}{d} \right)^2.
\end{align*}

\section*{Acknowledgments}

R.K. wants to thank Matthias Christandl and Fr\'ed\'eric Dupuis for introducing him to the concept of POVM norm constants.
This work has been supported by the Excellence Initiative of the German Federal and State Governments (Grant ZUK 81), the ARO under contract W911NF-14-1-0098 (Quantum Characterization, Verification, and Validation), and the DFG (SPP1798 CoSIP). 
Major parts of this project were undertaken while DG and RK participated in the \emph{Mathematics of Signal Processing} program of the Hausdorff Research Institute of Mathematics at the University of Bonn.
We thank Martin Kliesch for relevant last-minute inputs.

\bibliographystyle{ieeetr}
\bibliography{cliffinv,all_references}

\begin{widetext}

\section*{Appendix}

\subsection{Auxiliary statements for deriving \autoref{thm:4design} and \autoref{thm:main}}

With the notable exception of \cite{lancien_distinguishing_2013}, previous derivations \cite{ambainis_quantum_2007,matthews_distinguishability_2009} of the fourth moment bound presented in \eqref{eq:4design_bound} have assumed $X$ to be traceless. 
This additional assumption considerably simplifies the task at hand. Here, we prove a similar bound valid for arbitrary $X \in \mathcal{H}_d$ at the cost of a slightly larger multiplicative constant.
At the  heart of this derivation is \cite[Lemma 17]{kueng_low_2015} which provides a closed-form expression for the object at hand:

\begin{lemma}\label{lem:4mIneq}
	Suppose $X$ is a nonzero Hermitian operator and $y=|\tr(X)|/\|X\|_2$. Then 
\begin{align}\label{eq:4mIneq}
&\frac{24 \tr \left( P_{\mathrm{Sym}^4} X^{\otimes 4} \right)}{\left( \tr (X^2) + \tr (X)^2 \right)^2}
\leq 3+\frac{6+8y-2y^4}{(1+y^2)^2}\leq \frac{3}{5}(7+4\cdot 2^{1/3}+3\cdot 2^{2/3})\approx 10.08113.
\end{align}
Here the second inequality 
is saturated  iff $y= 2^{1/3}-1$;  the first one cannot be saturated except when $y=1$ and $X$ has rank 1, but it can be approached with arbitrarily small gap.
\end{lemma}
When $X$ is traceless, \autoref{lem:4mIneq} implies that
 \begin{align}
 &\frac{24 \tr \left( P_{\mathrm{Sym}^4} X^{\otimes 4} \right)}{\left( \tr (X^2) + \tr (X)^2 \right)^2}<9, 
 \end{align}
where the upper bound can be approached with arbitrarily small gap.
 
\begin{proof}
According to \cite[Lemma 17]{kueng_low_2015},	
\begin{align}
24 \tr \left( P_{\mathrm{Sym}^4} X^{\otimes 4} \right)
=&  \left( \tr (X)^4 + 8 \tr (X) \tr (X^3) + 3 \tr (X^2)^2 + 6 \tr (X)^2 \tr (X^2) +6 \tr (X^4) \right) \nonumber \\
=&  3 \left( \tr (X^2) + \tr (X)^2 \right)^2 +  8 \tr (X) \tr (X^3) + 6 \tr (X^4) - 2 \tr (X)^4 \nonumber \\
\leq &  
3 \left( \| X \|_2^2 + \tr (X)^2 \right)^2 +  8 | \tr (X)| \| X \|_3^3 + 6 \| X \|_4^4 - 2 \tr (X)^4\nonumber\\
\leq &  3 \left( \| X \|_2^2 + \tr (X)^2 \right)^2 +  8 | \tr (X)| \| X \|_2^3 + 6 \| X \|_2^4 - 2 \tr (X)^4
\label{aeq:4m1},
\end{align}
where the first inequality is saturated iff $X\geq0$ or $X\leq 0$, and the second one is saturated iff $\| X \|_4=\| X \|_3=\| X \|_2$, that is, $X$ has rank 1. Consequently, 
\begin{align}
\frac{24 \tr \left( P_{\mathrm{Sym}^4} X^{\otimes 4} \right)}{\left( \tr (X^2) + \tr (X)^2 \right)^2}
&\leq   3 + \frac{ 8 | \tr (X)| \| X \|_2^3 + 6 \| X \|_2^4 - 2 \tr (X)^4}{\left( \| X \|_2^2 + \tr (X)^2 \right)^2 }
=f(y):= 3+\frac{6+8y-2y^4}{(1+y^2)^2}\nonumber\\
&\leq \frac{3}{5}(7+4\cdot 2^{1/3}+3\cdot 2^{2/3})\approx 10.08113.
\end{align}
Here the first inequality is saturated iff $X$ has rank 1 (in which case $y=1$). To derive the second inequality, note that 
\begin{equation}
f'(y)=\frac{8(1-3y-3y^2-y^3)}{(1+y^2)^3},
\end{equation}
which is positive when $0\leq y< 2^{1/3}-1$ and negative when $ y> 2^{1/3}-1$. So the maximum of $f(y)$ for $y\geq0$
is attained when  $y= 2^{1/3}-1$, in which case 
\begin{equation}
f(2^{1/3}-1)= \frac{3}{5}(7+4\cdot 2^{1/3}+3\cdot 2^{2/3}).
\end{equation}
 
 Although the first inequality in \ref{eq:4mIneq} can not be saturated except when $y=1$, the bound can be approached arbitrarily close if we do not impose any restriction on the rank of $X$. To show this point, suppose $X=\diag(ak,-1,-1,\ldots,-1)$ has rank $k+1$, where $a$ is a real constant to be determined later.
 Then 
 \begin{align}
\tr(X)=k(a-1),\quad \|X\|_2^2=a^2 k^2+k,\quad \tr(X^3)=a^3k^3-k \quad \|X\|_4^4=a^4k^4+k.
 \end{align}
Assuming  $y\geq 0$, $y\neq1$, $k\geq y^2$, and let 
 \begin{align}
a=\frac{k+\sqrt{ky^2(1+k-y^2)}}{k(1-y^2)}.
\end{align}
Then $\tr (X) \tr (X^3)\geq0$, $|\tr(X)|/\|X\|_2=y$,
\begin{align}
\lim_{k\rightarrow\infty}a=\frac{1}{1-y},\quad \lim_{k\rightarrow\infty} \frac{|\tr(X^3)|}{\|X\|_2^3}=1,\quad \lim_{k\rightarrow\infty} \frac{\|X\|_4}{\|X\|_2}=1,
\end{align}
which implies that
\begin{align}
\lim_{k\rightarrow\infty} \frac{24 \tr \left( P_{\mathrm{Sym}^4} X^{\otimes 4} \right)}{\left( \tr (X^2) + \tr (X)^2 \right)^2}= 3+\frac{6+8y-2y^4}{(1+y^2)^2}.
\end{align}
\end{proof}

\begin{lemma}\label{lem:4mIneq2}
	Suppose $X$ is a nonzero Hermitian operator and $y=|\tr(X)|/\|X\|_2$. Then 
\begin{align}
\frac{24 \tr \left( P_{\mathrm{Sym}^4} X^{\otimes 4} \right)\tr(X^2)}{\left( \tr (X^2) + \tr (X)^2 \right)^2}
&\leq\frac{3(1+y^2)^2+6+8y-2y^4}{(1+y^2)^3}
< 9.673  \label{aeq:4m2}.
\end{align}
	Here the first inequality cannot be saturated except when $y=1$ and $X$ has rank 1, but it can be approached with arbitrarily small gap.
\end{lemma}
\begin{proof}
The lemma follows from \autoref{lem:4mIneq} except for the second inequality in \autoref{aeq:4m2}. To derive this inequality, let 
\begin{equation}
f(y)=\frac{3(1+y^2)^2+6+8y-2y^4}{(1+y^2)^3}; 
\end{equation}
then 
\begin{equation}
f'(y)=-\frac{2(-4+21y+20y^2+10y^3+y^5)}{(1+y^2)^4}.
\end{equation}
Note that $(1+y^2)^4f'(y)$ is monotonic decreasing with $y$ when $y\geq0$ and has a unique  real root $y_0>0$. Therefore, the maximum of $f(y)$ is attained when $y=y_0$. Now it is straightforward to verify that $f(y_0)<9.673$.  Calculation shows that 
\begin{equation}
y_0\approx 0.163078,\quad f(y_0)\approx 9.67249.
\end{equation}
\end{proof}

\begin{lemma}\label{lem:4mIneq3}
	Suppose $X$ is a rank-2 Hermitian operator. Then 
	\begin{align}\label{eq:4mIneq3}
	\frac{24\| X \|_1^2\tr \left( P_{\mathrm{Sym}^4} X^{\otimes 4} \right)}{ [\| X \|_2^2 + \tr (X)^2]^3}\leq \frac{5}{81}(95+32\sqrt{10})\approx 12.1107.
	\end{align}
	If $X$ is in addition traceless, then 
		\begin{align}\label{eq:4mIneq32}
		\frac{24\| X \|_1^2\tr \left( P_{\mathrm{Sym}^4} X^{\otimes 4} \right)}{ [\| X \|_2^2 + \tr (X)^2]^3}=12.
		\end{align}
\end{lemma}
\begin{proof}
Note that the left hand side of \eqref{eq:4mIneq4} is invariant when $X$ is multiplied by any nonzero real constant. Without loss of generality, we may assume that the two nonzero eigenvalues of  
$X$ are equal to $1,x $ with $-1\leq x\leq 1$. Then \begin{align}
&\|X\|_1=1+|x|, \quad \|X\|_2=1+x^2, \quad  \tr \left( P_{\mathrm{Sym}^4} X^{\otimes 4} \right)=1+x+x^2+x^3+x^4,
\end{align}
so that 
\begin{align}\label{aeq:rank2}
\frac{24\| X \|_1^2\tr \left( P_{\mathrm{Sym}^4} X^{\otimes 4} \right)}{ [\| X \|_2^2 + \tr (X)^2]^3}
=f(x):=\frac{3(1+|x|)^2(1+x+x^2+x^3+x^4) }{(1+x+x^2)^3}.
\end{align}
If  $x\geq0$, then  $f(x)\leq 3$ according to the following equation,
\begin{equation}
(1+|x|)^2(1+x+x^2+x^3+x^4)-(1+x+x^2)^3=-x^2(2+3x+2x^2)\leq 0.
\end{equation}
If  $-1\leq x< 0$, then 
\begin{align*} 
f(x):=\frac{3(1-x)^2(1+x+x^2+x^3+x^4) }{(1+x+x^2)^3},\quad f'(x)=\frac{3(-1+x)(1+x)(4+4x-x^2+4x^3+4x^4)}{(1+x+x^2)^4}.
\end{align*}
Let $x_0$ be the unique real root of  $4+4x-x^2+4x^3+4x^4$ which lies between $-1$ and 0, then $f'(x)\geq 0$ if $-1\leq x\leq x_0$ and $f'(x)\leq 0$ if $x_0\leq x\leq 0$. Therefore, the maximum of $f(x)$ is attained when $x=x_0$, in which case
\begin{equation}
f(x_0)=\frac{5}{81}(95+32\sqrt{10}).
\end{equation}

If $X$ is in addition traceless, then $x=-1$, so \eqref{eq:4mIneq32} follows from \autoref{aeq:rank2}.
\end{proof}

\begin{lemma}\label{lem:4mIneq4}
	Suppose $X$ is a rank-2 Hermitian operator. Then 
	\begin{align}\label{eq:4mIneq4}
\frac{6 \| X \|_1^4 \| X \|_2^2+24\| X \|_1^2\tr \left( P_{\mathrm{Sym}^4} X^{\otimes 4} \right)}{ [\| X \|_2^2 + \tr (X)^2]^3}\leq 36,
	\end{align}
where the upper bound is saturated iff $X$ is traceless.
\end{lemma}
\begin{proof}
As in the proof of \autoref{lem:4mIneq3}, we may assume that the two nonzero eigenvalues of 
$X$ are equal to $1,x $ with $-1\leq x\leq 1$. Then 
\begin{align*}
\frac{6 \| X \|_1^4\| X \|_2^2+24\| X \|_1^2\tr \left( P_{\mathrm{Sym}^4} X^{\otimes 4} \right)}{ [\| X \|_2^2 + \tr (X)^2]^3}
&=f(x):=\frac{6(1+|x|)^4(1+x^2)+24(1+|x|)^2(1+x+x^2+x^3+x^4) }{ 8(1+x+x^2)^3}.
\end{align*}
When $x\geq0$, it is straightforward to verify that $f(x)\leq 9$. When $-1\leq x< 0$,
\begin{align*} f(x)=&\frac{6(1-x)^4(1+x^2)+24(1+x+x^2+x^3+x^4) }{ 8(1+x+x^2)^3}=
\frac{3(1-x)^2(5+2x+6x^2+2x^3+5x^4) }{ 4(1+x+x^2)^3},
\end{align*}
whose derivative is given by 
\begin{align*} f'(x)=&
\frac{3(-23+9x^2-9x^4+23x^6) }{ 4(1+x+x^2)^4}\leq0,
\end{align*}
Therefore, $f(x)\leq f(-1)=36$, and the upper bound is saturated iff $x=-1$, in which case $X$ is traceless. 
\end{proof}

\subsection{Fourth moment implications of \autoref{thm:main_technical} for stabilizer states}

\begin{lemma}
Fix $d=2^n$, let $\phi = | \phi \rangle \! \langle \phi| \in S_d$ be any pure state and $\left\{|x_k \rangle \right\}_{k=1}^N \subseteq \mathbb{C}^d$ denotes the set of all stabilizer states. Then the random variable $S_\phi = \langle x_k | \phi |x_k \rangle$ with probability $\frac{1}{N}$ obeys
\begin{equation*}
\mathbb{E} \left[ S_\phi^4 \right] \leq \frac{30}{(d+4)(d+2)(d+1)d}.
\end{equation*}
\end{lemma}

\begin{proof}
The set of all stabilizer states forms a Clifford orbit and \autoref{thm:main_technical} implies
\begin{align}
\mathbb{E} \left[ S_\phi^4 \right] =& \frac{1}{N} \sum_{k=1}^N \langle x_k | \phi |x_k \rangle^4 = \mathrm{tr} \left( \frac{1}{N} \sum_{k=1}^N \left( |x_k \rangle \! \langle x_k | \right)^{\otimes 4} \phi^{\otimes 4} \right) \nonumber\\
= &d \binom{d+2}{3}^{-1} \left( \left( \alpha (z)-\beta (z) \right) \mathrm{tr} \left( Q \phi^{\otimes 4}  P_{\mathrm{Sym}^4} \right) + \beta \mathrm{tr} \left( P_{\mathrm{Sym}^4} \phi^{\otimes 4} \right) \right), \label{eq:app_aux1}
\end{align}
with $Q = \frac{1}{d^2} \sum_{k=1}^{d^2} W_k^{\otimes 4}$ and $\alpha(z),\beta (z) \in \mathbb{R}$ depend on the choice of fiducial. According to Sec.~\ref{sec:fiducials}  we obtain
\begin{equation*}
\alpha (z) = \frac{1}{d} \quad \textrm{and} \quad \beta (z) = 4 \frac{1-\alpha (z)}{(d+4)(d-1)} = \frac{4}{(d+4)d}.
\end{equation*}
for stabilizer state fiducials. This in turn implies
\begin{equation*}
\alpha (z) - \beta (z) = \frac{1}{d} - \frac{4}{(d+4)d} = \frac{1}{d+4}
\end{equation*}
which confirms the upper bound presented in \eqref{eq:alphabound2}, because the difference between $\alpha(z)$ and $\alpha (z)$ is maximal for stabilizer state fiducials.
Moreover, tensor products $\phi^{\otimes k}$ of pure states are always contained in the totally symmetric subspace: $P_{\mathrm{Sym}^4} \phi^{\otimes 4} = \phi^{\otimes 4} P_{\mathrm{Sym}^4} = \phi^{\otimes 4}$.
These relations allow us to considerably simplify \eqref{eq:app_aux1}:
\begin{align}
\mathbb{E} \left[ S_\phi^4 \right]
=& d \binom{d+2}{3}^{-1} \left( \frac{1}{d+4} \mathrm{tr} \left( Q \phi^{\otimes 4} \right) + \frac{4}{(d+4)d} \mathrm{tr} \left( \phi^{\otimes 4} \right) \right) 
= \frac{d}{d+4} \binom{d+2}{3}^{-1} \left( \frac{1}{d^2} \sum_{k=1}^{d^2} \mathrm{tr} \left( W_k \phi \right)^4 + \frac{4}{d} \right) \nonumber\\
=& \frac{6}{(d+4)(d+2)(d+1)} \left( \frac{1}{d^2}\| \Xi ( \phi ) \|_{\ell_4}^4 + \frac{4}{d} \right). \label{eq:app_aux2}
\end{align}
The expression $\| \Xi (\phi) \|_{\ell_4}^4$ is maximized, if $\phi = | \phi \rangle \! \langle \phi|$ is itself a stabilizer state: $\| \Xi (\phi)\|_{\ell_4}^4 \leq d$. Inserting this tight upper bound into \eqref{eq:app_aux2} implies the claim.
\end{proof}

\subsection{Proof of \autoref{lem:effective_rank}}

This statement implies that the effective rank $r_{\mathrm{eff}}(X) = \frac{ \| X \|_1^2}{\| X \|_2^2}$ of any matrix of the form $X = \rho - \frac{1}{d} \mathbb{I}$, $\rho \in \mathcal{S}_d$ is proportional to the rank of $\rho$.
In order to show this, we start by computing the Hilbert-Schmidt norm of $X$:
\begin{align*}
\| X \|_2^2 =& \tr \left( \rho^2 \right) + \frac{1}{d^2} \tr (\mathbb{I}) = \tr \left( \rho^2 \right)-\frac{1}{d}.
\end{align*}
Recall that the minimal purity of any rank-$r$ state $\rho$ is $\mathrm{tr}(\rho^2)=\frac{1}{r}$ which in turn implies
\begin{equation}
\| X \|_2^2 \geq  \frac{d-r}{dr}. \label{eq:effective_rank_aux1}
\end{equation}
For computing the trace norm, we employ an eigenvalue decomposition $\rho = \sum_{k=1}^r \lambda_k |k \rangle \! \langle k|$ of $\rho$ and in turn write $\mathbb{I} = \sum_{k=1}^d |k \rangle\! \langle k|$. Consequently
\begin{align*}
\| X \|_1 = \sum_{k=1}^r \left| \lambda_k - \frac{1}{d} \right| + \sum_{k=r+1}^d \frac{1}{d} 
\leq \sqrt{r \sum_{k=1}^r \left( \lambda_k - \frac{1}{d} \right)^2}+\frac{d+r}{d},
\end{align*}
because $\| x \|_{\ell_1} \leq \sqrt{r} \| x \|_{\ell_2}$ for any $x \in \mathbb{C}^r$. Applying $\sum_{k=1}^r \lambda_k^2 = \tr (\rho^2)$, $\sum_{k=1}^r \lambda_k = \tr (\rho) =1$ and resorting to \eqref{eq:effective_rank_aux1} we obtain
\begin{align*}
\| X \|_1 \leq & \sqrt{r \sum_{k=1}^r \left( \lambda_k - \frac{1}{d} \right)^2}+\frac{d+r}{d} 
= \sqrt{ r \left( \tr (\rho^2) - \frac{1}{d} - \frac{d-r}{d} \right)}+\frac{d-r}{d^2}  \\
=& \sqrt{ r \left( \|X \|_2^2 - \frac{r}{d}\frac{d-r}{rd} \right)} + \sqrt{r} \sqrt{\frac{d-r}{d}} \sqrt{\frac{d-r}{dr}} 
\leq  \sqrt{r \left( 1- \frac{r}{d} \right) \| X \|_2^2} + \sqrt{r \frac{d-r}{d}} \| X \|_2 \\
=& 2 \sqrt{r\frac{d-r}{d}}\| X\|_2.
\end{align*}
Combining these two relations implies
\begin{equation*}
r_{\mathrm{eff}}(X) = \frac{ \| X\|_1^2}{\|X\|_2^2} = \frac{4r(d-r)}{d},
\end{equation*}
as claimed. The second bound follows from the fact that $\max \left\{r,d-r \right\}\leq \frac{d-1}{d}\leq d-1$ for any $1 \leq r \leq d-1$ (the case $r=d$ is trivial, because it implies $X=0$). Consequently:
\begin{align*}
\frac{4r(d-r)}{d} =& \frac{4}{d} \max \left\{ r, d-r \right\} \min \left\{ r,d-r \right\} 
\leq 4\frac{d-1}{d} \min \left\{ r, d-r \right\}.
\end{align*}
The fact that both bounds are saturated, follows from a straightforward computation for $\rho = \sum_{k=1}^r |k \rangle \! \langle k|$ (first bound) and then setting $r=1$ and $r=d-1$, respectively (second bound).

\end{widetext}
\end{document}